\documentclass{article}
\usepackage[utf8]{inputenc}

\usepackage{amsmath, amsthm, amssymb}
\usepackage{mathabx}
\usepackage{mathtools}
\usepackage{graphicx}
\usepackage{natbib}
\usepackage{caption}
\usepackage{subcaption}
\usepackage{enumerate}

\newcommand{\supplementname}{Supplement}

\usepackage{hyperref}

\usepackage[margin=1.45in]{geometry}

\hypersetup{hyperfootnotes=false,colorlinks,citecolor=blue,urlcolor=blue,linkcolor=blue}

\usepackage[shortlabels]{enumitem}

\usepackage[linesnumbered,ruled,resetcount]{algorithm2e}
\DontPrintSemicolon
\SetAlFnt{\normalsize}
\SetKwInOut{Input}{Input}
\newcommand{\nextnr}{\stepcounter{AlgoLine}\ShowLn}

\newcommand{\hrulealg}[0]{\vspace{1mm} \hrule \vspace{1mm}}

\usepackage{xparse}
\let\oldeqref\eqref
\makeatletter
\RenewDocumentCommand\eqref{s m}{%
  \IfBooleanTF#1%
  {\oldeqref{#2}}
  {\oldeqref{#2}}
}
\makeatother

\usepackage{nikos_tex}

\usepackage[hang,flushmargin]{footmisc}

\usepackage{changes}

\graphicspath{{./figures/}}

\theoremstyle{plain}
\newtheorem{prop}{Proposition}
\newtheorem{theo}[prop]{Theorem}

\newtheorem{coro}[prop]{Corollary}

\theoremstyle{definition}
\newtheorem{exam}{Example}

\newtheorem{assum}{Assumption}

\theoremstyle{remark}
\newtheorem{rema}{Remark}

\newcommand{\authorblock}[1]{\begin{tabular}{@{}c@{}}#1\end{tabular}}

\title{Empirical Bayes mean estimation with nonparametric errors via order statistic regression on replicated data}

\date{\today}

\author{ 
  \begin{tabular}{c@{\qquad}c} 
  \authorblock{Nikolaos Ignatiadis\thanks{This work was done as part of an internship in Google Ads.}\\
                \texttt{ignat@stanford.edu}} & 
  \authorblock{Sujayam Saha \\ \texttt{sujayam@google.com}} \\
  \\
  \authorblock{Dennis L. Sun \\ \texttt{dsun09@calpoly.edu}} &
  \authorblock{Omkar Muralidharan \\ \texttt{omuralidharan@google.com}} \\
  & 
  \end{tabular}  
}

\begin{document}

\maketitle

\begin{abstract}
We study empirical Bayes estimation of the effect sizes of $\n$ units 
from $\B$ noisy observations on each unit. We show that it is possible to achieve near-Bayes optimal mean squared error, without any assumptions or knowledge about the effect size distribution or the noise. The noise distribution can be heteroskedastic and vary arbitrarily from unit to unit. Our proposal, which we call Aurora, leverages the replication inherent in the $\B$ observations per unit and recasts the effect size estimation problem as a general regression problem. 
Aurora with linear regression provably matches the performance of a wide array of estimators including the sample mean, the trimmed mean, the sample median, as well as James-Stein shrunk versions thereof. Aurora automates effect size estimation for Internet-scale datasets, as we demonstrate on data from a large technology firm.\\
\newline
\noindent%
\end{abstract}

\section{Introduction}

Empirical Bayes (EB)~\citep{efron2012large,robbins1956empirical,robbins1964empirical} and related shrinkage methods are the {\em de facto} standard for estimating effect sizes in many disciplines. In genomics, EB 
is used to detect differentially expressed genes when the number of samples is small~\citep{smyth2004linear, love2014moderated}. In survey sampling, EB improves noisy estimates of quantities, like the 
average income, for small communities~\citep{rao2015small}. 
The key insight of EB is that one can often estimate unit-level quantities better by sharing information across units, rather than analyzing each unit separately.

Formally, EB models the observed data $\boldZ = (Z_1, ..., Z_\n)$ as arising from the following generative process:
\begin{align}
\mu_i &\sim G, &Z_i \mid \mu_i &\sim F(\cdot \mid \mu_i), & i &= 1, ..., \n.
\label{eq:basic_eb}
\end{align}
The goal here is to estimate the mean parameters,  
$\mu_i := \EE[F]{Z_i \mid \mu_i}$ for $i=1, ..., \n$,
from the observed data 
$\boldZ$. If $G$ and $F$ are fully 
specified, then the optimal estimator (in the sense of mean squared error) 
is the posterior mean $\EE[G, F]{ \mu_i \cond Z_i}$, 
which achieves the Bayes risk. Empirical Bayes deals with the case where $F$ or $G$ is unknown, so the Bayes rule cannot be calculated. Most modern EB methods~\citep{jiang2009general, brown2009nonparametric, muralidharan2012high, saha2017nonparametric} assume that $F$ is known (say, $F(\cdot \cond \mu_i) = \mathcal{N}(\mu_i,1)$) and construct estimators $\hat{\mu}_i$ that asymptotically match the risk of the unknown Bayes rule, without making any assumptions about the unknown prior $G$. 

We examine the same problem of estimating the $\mu_i$s 
when the likelihood $F$ is also unknown. Indeed, knowledge of $F$ is an assumption that requires substantial domain expertise. For example, it took many years for 
the genomics community to agree on an EB model for detecting differences in gene expression based on microarray data~\citep{baldi2001bayesian, lonnstedt2002replicated, smyth2004linear}. Then, once this technology was superseded by bulk RNA-Seq, the community had to 
devise a new model from scratch, 
eventually settling on the negative binomial likelihood~\citep{love2014moderated, gierlinski2015statistical}.

Unfortunately, there is no way to avoid making such strong assumptions when there is no 
information besides the one $Z_i$ 
per $\mu_i$. If $F$ is 
even slightly
underspecified, then it becomes hopeless to disentangle $F$ 
from $G$. To appreciate the problem, consider the
Normal-Normal model:
\begin{align}
    G &= \mathcal{N}(0, A) & F(\cdot \mid \mu_i) &= 
\mathcal{N}(\mu_i, \sigma^2).
\label{eq:normal_model_unknown_variance}
\end{align}
Here, $Z_i$ is marginally distributed as
$\mathcal{N}(0, A + \sigma^2)$, and the observations $Z_i$ only provide information about $A+\sigma^2$. Now, when $\sigma^2$ is known, $A$ can 
be estimated by first estimating the marginal variance and subtracting $\sigma^2$. Indeed,~\citet{efron1973stein} showed that by plugging in a particular estimate of $A$ into the Bayes rule $\EE[G, F]{\mu_i \cond Z_i} = (1 - \frac{\sigma^2}{\sigma^2 + A}) Z_i$, one recovers the celebrated James-Stein estimator~\citep{james1961estimation}. Yet, as soon as $\sigma^2$ is unknown, then $A$
(and hence, $G$) is unidentified, and there is no hope of approximating the unknown Bayes rule.

However, as any student of random effects knows, 
the Normal-Normal model~\eqref{eq:normal_model_unknown_variance} 
becomes identifiable if 
we simply have independent replicates $Z_{ij}$ for 
each unit $i$. The driving force behind this work is an analogous observation in the context of empirical Bayes estimation: replication makes it possible to estimate $\mu_i$ with no assumptions on $F$ or $G$ whatsoever. The method we propose, described in the next 
section, performs well in practice and nearly matches the risk of 
the Bayes rule, which depends on the unknown $F$ and $G$.

\section{The Aurora Method}
\label{sec:method}

First, we formally specify the EB model when 
replicates $Z_{ij} \in \RR$ are available. 
\begin{align}
\label{eq:universal_eb}
(\mu_i, \alpha_i) &\simiid G, & i &= 1, \dotsc, \n \nonumber \\ 
Z_{ij} \mid (\mu_i, \alpha_i) & \simiid F(\cdot \mid \mu_i, \alpha_i) & j&= 1,\dotsc, \B.
\end{align}
Again, the quantity of interest is the mean parameter 
$\mu_i := \EE[F]{Z_{ij} \mid \mu_i, \alpha_i}$. The additional parameter 
$\alpha_i$ is a nuisance parameter that allows 
for heterogeneity across the units, while 
preserving exchangeability~\citep{galvao2014estimation,okui2019kernel}.
For example, $\alpha_i$ is commonly taken to be
the conditional variance 
$\sigma_i^2 := \Var[F]{Z_{ij} \mid \mu_i, \alpha_i}$ to allow for heteroskedasticity. 
However, $\alpha_i$ could even be infinite-dimensional---for instance, a random 
element from 
a space of distributions. The $\alpha_i$ have no impact on our 
estimation strategy and are purely a technical device.

Given data from model~\eqref{eq:universal_eb}, one approach 
would be to collapse the replicates into a single observation per unit---say, by taking their mean---which would bring us back to the setting of model~\eqref{eq:basic_eb}. We could appeal to 
the Central Limit Theorem to justify knowing that 
the likelihood is Normal. An important message of this paper is that we can do \emph{better} by using the replicates.

\subsection{Proposed Method}
Since we have replicates, we can split the $Z_{ij}$ 
into two groups for each $i$. First, consider 
the case of $\B = 2$ replicates so that we can write $(X_i, Y_i)$ for $(Z_{i1}, Z_{i2})$. 
Now, $X_i$ and $Y_i$ are 
conditionally independent given $(\mu_i, \alpha_i)$. Figure~\ref{fig:two_replicates} illustrates the relationship 
between $X_i$ and $Y_i$ under two different settings. 
The key insight is that the conditional mean 
$\EE[G, F]{Y_i \cond X_i}$ is (almost surely) identical to the posterior 
mean $\EE[G, F]{\mu_i \cond X_i}$, by the 
following elementary calculation~\citep{krutchkoff1967supplementary}. 
(For convenience, 
we suppress the dependence of the expected values on $G,F$.)
\begin{equation}
\label{eq:eb_to_reg}
\EE{Y_i \cond X_i} = \EE{\EE{Y_i \cond \mu_i, \alpha_i, X_i} \cond X_i} = \EE{\EE{Y_i \cond \mu_i, \alpha_i} \cond X_i} = \EE{ \mu_i \cond X_i}.
\end{equation}
This suggests that we can estimate the 
Bayes rule based on $X_i$ (i.e., the posterior mean $\EEInline{\mu_i \cond X_i}$) by simply 
regressing $Y_i$ on $X_i$ using any black-box predictive model, 
such as a local averaging smoother. Let $\hat{m}(\cdot)$ be 
the fitted regression 
function; our estimate of each $\mu_i$ is then just 
$\hat{\mu}_i = \hat{m}(X_i)$. 

\begin{figure}
    \centering
    \includegraphics[width=\linewidth]{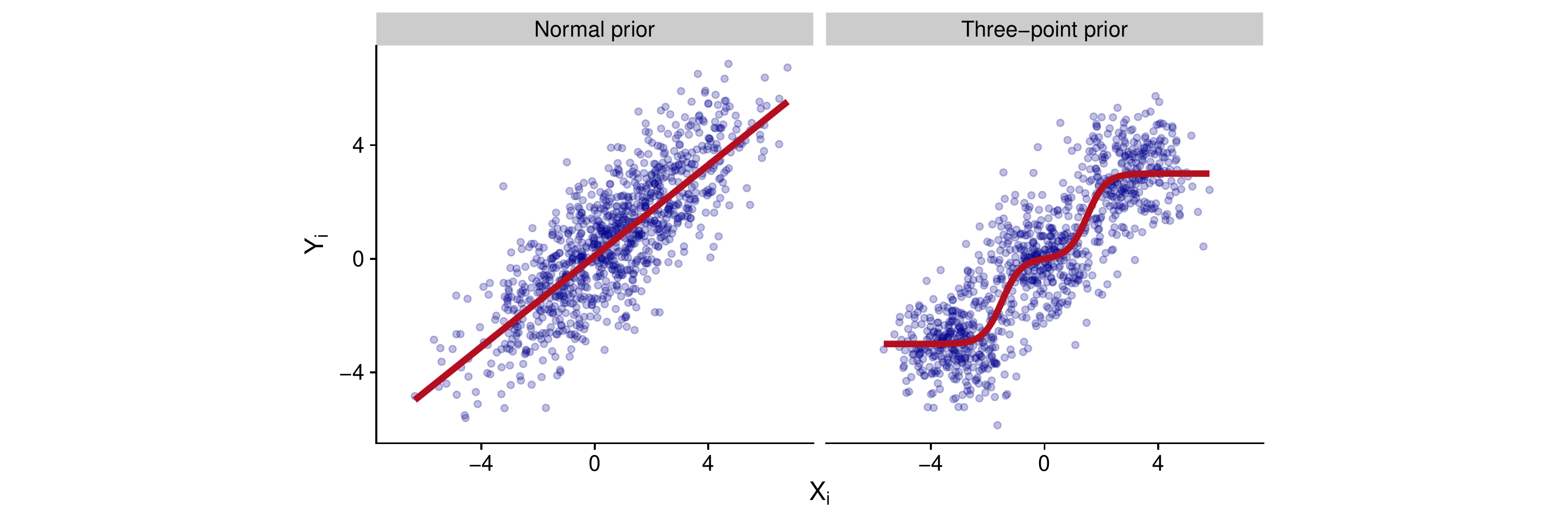}
    \caption{{\small\textbf{Empirical Bayes with replicates:}
    Two simulations with $\n = 1000$ and $\B = 2$. 
    First, we draw $\mu_i \sim G$, where $G= \mathcal{N}(0.5,4)$ in the left panel and $G$ is the uniform distribution on the discrete set $\cb{-3,\; 0,\; 3}$ in the right panel. Then, for each $i$, we draw $X_i, Y_i \cond \mu_i \simiid \mathcal{N}(\mu_i,1)$ and plot the points $(X_i, Y_i)$. The line shows the posterior mean $\EE{\mu_i \cond X_i}$, which in light of~\eqref{eq:eb_to_reg}, is identical to the conditional mean 
    $\EE{Y_i \cond X_i}$.}}
    \label{fig:two_replicates}
\end{figure}

To extend this method to $\B > 2$, we can again split the replicates $\boldZ_i$ into 
two parts: 
\begin{align}
    \boldX_i &:= \boldX_i(j) := (Z_{i1},\dotsc, Z_{i(j-1)}, Z_{i(j+1)},\dotsc, Z_{i\B}) & Y_i &:= Y_i(j) := Z_{ij},
    \label{eq:xy_def}
\end{align}
where $j \in \cb{1,\dotsc,\B}$ is an arbitrary fixed index for now. (We suppress $j$ in the notation below.)
Now, one approach is to summarize the vector $\boldX_i$ by 
the mean of its values $\bar X_i$ and regress $Y_i$ on $\bar X_i$ 
to learn $\EE{ \mu_i \cond \widebar{X}_i}$
\citep{coey2019improving}. This works for essentially 
the same reason as the $\B=2$ case. However, unless $\bar{X}_i$ 
is sufficient for $(\mu_i, \alpha_i)$ in 
model~\eqref{eq:universal_eb}, then 
$\EE{ \mu_i \cond \widebar{X}_i}$ will be different
from and suboptimal to 
$\EE{ \mu_i \cond \boldX_i}$.

Instead, we propose learning
$\EE{ \mu_i \cond \boldX_i}$ directly. 
The rationale is contained in the following 
result.

\begin{prop}
\label{prop:main}
Let $Z_{ij}$ be generated according to \eqref{eq:universal_eb} and assume that $\EEInline{\abs{\mu_i}} < \infty$, $\EEInline{\abs{Z_{ij}}} < \infty$. Define $\boldX_i$ and $Y_i$ as in \eqref{eq:xy_def}. 
Let $\ordX_i$ be the vector of order statistics of $\boldX_i$:
$$\ordX_i := \p{X_{i}^{(1)}, \dotsc, X_{i}^{(\B-1)}},\;\;\;\;\;\; X_{i}^{(1)} \leq \dotsc \leq X_{i}^{(\B-1)}.$$
That is, $\ordX_i$ is simply a sorted version of $\boldX_i$. Then, almost surely
\begin{equation}
\label{eq:eb_to_reg2}
\EE{Y_i \mid \ordX_i} = \EE{\mu_i \mid \boldX_i}.
\end{equation}
\end{prop}

\begin{proof}
The same argument as \eqref{eq:eb_to_reg} shows that \smash{$\EEInline{Y_i \mid \ordX_i} = \EEInline{ \mu_i \mid \ordX_i}$} almost surely. Now, under exchangeable sampling, the order statistics $\ordX_i$ are sufficient for $(\mu_i, \alpha_i)$ and therefore it follows that \smash{$\EEInline{\mu_i \mid \ordX_i} = \EEInline{\mu_i \mid \boldX_i}$} almost surely, with no assumptions on $F$.
\end{proof}

Equation~\eqref{eq:eb_to_reg2} suggests that we should 
regress $Y_i$ on the order statistics $\ordX_i$ to learn 
a function \smash{$\hat{m}(\cdot):\RR^{\B-1}\to \RR$} that approximates the posterior mean. The fitted values \smash{$\hat{m}(\ordX_i)$}  can be used to estimate $\mu_i$.
There is one more detail worth mentioning. The estimate  \smash{$\hat{m}(\ordX_i)$} depends on the arbitrary choice of $j$ in~\eqref{eq:xy_def}. To reduce the variance of the estimate, we average over all choices of $j \in \cb{1,\dotsc,\B}$. This method, 
summarized in Table~\ref{tab:algorithm}, is called Aurora, which 
stands for ``Averages of Units by Regressing on Ordered Replicates Adaptively.''

\begin{table}
    \centering
    \begin{tabular}{cp{5in}}
    \hline
    \multicolumn{2}{l}{Aurora: ``Averages of Units by Regressing on Ordered Replicates Adaptively.''}   \\
    \hline
    \multicolumn{2}{l}{For $j \in \cb{1,\dotsc,\B}$}\\
    \;\;\;\;1. & Split the replicates for each unit, $\boldZ_i$,
    into $\boldX_i := (Z_{i1},\dotsc, Z_{i(j-1)}, Z_{i(j+1)}, \dotsc, Z_{i\B})$ and 
    $Y_i := Z_{ij}$, as in \eqref{eq:xy_def}. \\
    \;\;\;\;2. & For each $\boldX_i$, order the values to obtain 
    $\ordX_i$. \\
    \;\;\;\;3. & Regress $Y_i$ on $\ordX_i$ using any black-box 
    predictive model. Let $\hat m_j$ be the fitted regression function.\\
    \;\;\;\;4. &  Let $\muebds_{i,j} := \hat m_j(\ordX_i)$. \\
    \multicolumn{2}{l}{end}\\
    \multicolumn{2}{l}{Estimate each $\mu_i$ by $\muebds_i := \frac{1}{\B}\sum_{j=1}^{\B} \muebds_{i,j}$.}\\
    \hline
    \end{tabular}
    \caption{{\small \textbf{A summary of Aurora}, which is the proposed method for 
    estimating the means $\mu_i$ when the data come from 
    model~\eqref{eq:universal_eb}.}}
    \label{tab:algorithm}
\end{table}

\section{Related work}
The Aurora method is  closely related to the extensive literature on empirical Bayes, which we cite throughout this paper. One work that is worth emphasizing is~\citet{stigler1990}, who motivated empirical 
Bayes estimators, like James-Stein, 
through the lens of regression to the mean; 
for us, regression is not just a motivation but the estimation 
strategy itself. We were surprised 
to find a similar idea to Aurora in a forgotten 
manuscript, uncited to date, that 
\citet{johns1986fully} contributed to a symposium for Herbert Robbins
\citep{van1986adaptive}. Although \citet{johns1986fully} used a 
fairly complex predictive model (projection pursuit 
regression~\citep{friedman1981projection}),  we show theoretically (Sections~\ref{sec:properties},\ref{sec:linear_ebds}) and empirically (Sections~\ref{sec:empirical_simulations},\ref{subsec:cpc}) that the
order statistics encode enough structure 
that linear regression and $k$-nearest neighbor regression can be used as the predictive model.

Models similar to~\eqref{eq:universal_eb} with replicated noisy measurements of unobservable random quantities have been studied in the context of deconvolution and error-in-variables regression~\citep{devanarayan2002empirical,schennach2004estimation}. In econometrics, 
panel data with random effects are often modelled as in~\eqref{eq:universal_eb} with the additional potential complication of time dependence, i.e., $j=1,\dotsc,\B$ indexes time, while $i=1,\dotsc,\n$ may correspond to different geographic regions~\citep*{horowitz1996semiparametric, hall2003inference, neumann2007deconvolution, jochmans2018inference}. \citet{fithian2017family} use observations from model~\eqref{eq:universal_eb} to learn a low-dimensional smooth 
parametric model that describes the data well. However, their
goal is testing, while Aurora is geared for mean estimation.

The Aurora method is also related to a recent line of research that leverages black-box prediction methods to solve statistical tasks that are \emph{not} predictive in nature. For example, \citet{chernozhukov2017double} consider inference for low dimensional causal quantities when high dimensional nuisance components are estimated by machine learning. \citet{boca2018direct} reinterpret the multiple testing problem in the presence of informative covariates as a regression problem and estimate the proportion of null hypotheses conditionally on the covariates. The estimated conditional proportion of null hypotheses may then be used for downstream multiple testing methods~\citep{ignatiadis2018covariate}. Black-box regression models can also be used to improve empirical Bayes point estimates in the presence of side-information~\citep{ignatiadis2019covariate}. 

Finally, a crucial ingredient in Aurora is data splitting, which is
a classical idea in statistics~\citep{cox1975note}, 
typically used to ensure 
honest inference for low dimensional parameters. In the context of simultaneous inference,~\citet*{rubin2006method} and \citet*{ habiger2014compound} use data-splitting to improve power in multiple testing.

\section{Properties of the Aurora estimator}
\label{sec:properties}

We provide theoretical guarantees for the general
Aurora estimator described in Section~\ref{sec:method}. Throughout this section we split $\boldZ_i$ into $\boldX_i(j)$ and $Y_i(j)$ as in~\eqref{eq:xy_def}, and we write $\ordX_i(j)$ for the order statistics of $\boldX_i(j)$. We also write $\boldZ$ and $\ordX(j)$ for the 
concatenation of all the $\boldZ_i$ and $\ordX_i(j)$,
respectively. As before, we omit $j$ whenever a definition does not depend on the specific choice of $j$.
\subsection{Three oracle benchmarks}
In this section we define three benchmarks for assessing the quality of a mean estimator in model~\eqref{eq:universal_eb} (in terms of mean squared error). These benchmarks provide the required context to interpret the theoretical guarantees of Aurora that we establish in Section~\ref{subsec:regret_aurora}.

First, it is impossible to improve on the Bayes rule, 
so the Bayes risk serves as an oracle. 
We denote the Bayes risk (based on all 
$\B$ replicates) by
\begin{align}
\bayesRisk{\B}{G, F} &:= \EEInline[G,F]{ (\mu_i - \EEInline[G,F]{\mu_i \cond \boldZ_i})^2} \\
&= \EEInline[G,F]{ (\mu_i - \EEInline[G,F]{\mu_i \cond \ordX_i, Y_i})^2}. \nonumber
\end{align}
As explained in Proposition~\ref{prop:main}, the function we seek to mimic (for each $j$ in the loop of the Aurora algorithm in Table~\ref{tab:algorithm})  is the oracle Bayes rule based on the order statistics of $\B-1$ replicates,
\begin{equation}
\label{eq:bayesrule_minus_one}
     m^*(\ordX_i) := \EE[G,F]{\mu_i \cond \ordX_i}.
\end{equation}
The risk of $m^*(\cdot)$ is
the Bayes risk based on $\B - 1$ replicates,
\begin{align}
\bayesRisk{\B - 1}{G, F} &:= \EEInline[G,F]{ (\mu_i - \EEInline[G, F]{\mu_i \cond \ordX_i})^2} = \EEInline[G,F]{ (\mu_i - m^*(\ordX_i))^2}.
\label{eq:bayes_estimator_risk}
\end{align}
In the Aurora algorithm we average over the choice of held-out replicate $j$. Thus, we also define an oracle rule that averages $m^*(\ordX_i(j))$ over all $j$. We use the following notation for this oracle and its risk:
\begin{equation}
\label{eq:avgbayesrule}
\AvgBayesRule(\boldZ_i) = \frac{1}{\B} \sum_{j=1}^{\B}m^*\p{\ordX_i(j)},\;\;\AvgBayesRisk{\B-1 }{G, F} = \EE[G,F]{ \p{\mu_i - \AvgBayesRule(\boldZ_i)}^2}.
\end{equation}
It is immediate that $ \bayesRisk{\B}{G,F} \leq \AvgBayesRisk{\B-1}{G, F}$.  Also, by the definition of $\AvgBayesRule$ in~\eqref{eq:avgbayesrule} and by Jensen's inequality, it holds that,
$$
\begin{aligned}
\EE[G,F]{ \p{\mu_i - \AvgBayesRule(\boldZ_i)}^2}\, &= \,\EE[G,F]{\cb{\frac{1}{\B}\sum_{j=1}^{\B} \p{\mu_i - m^*\p{\ordX_i(j)}}}^2} \\
&\leq\, \frac{1}{\B} \sum_{j=1}^{\B} \EE[G,F]{ \p{\mu_i - m^*\p{\ordX_i(j)}}^2},
\end{aligned}
$$
\color{black}
so $\AvgBayesRisk{\B-1 }{G, F} \leq \bayesRisk{\B-1}{G,F}$.
This bound can be improved through the following insight: $m^*(\ordX_i(j)),\, j=1,\dotsc,\B$ are jackknife estimates of the posterior mean $\EE{ \mu_i \mid \boldZ_i}$ and $\AvgBayesRule(\boldZ_i)$ is their average. By a fundamental result for the jackknife~\citep[Theorem 2]{efron1981jackknife}, it holds that\footnote{Here we use the fact that the entries of $\boldZ_i$ are independent conditionally on $\mu_i, \alpha_i $. This inequality is also known in the theory of U-statistics~\citep[Theorem 5.2]{hoeffding1948class}.} $\VarInline{\AvgBayesRule(\boldZ_i) \mid \mu_i, \alpha_i} \leq \VarInline{   m^*(\ordX_i) \mid \mu_i, \alpha_i} \cdot (\B-1)/\B$. Armed with this insight, in \supplementname~\ref{subsec:proof_jackknife}, we prove that:

\begin{prop} 
\label{prop:averaged_jackknife}
Under model~\eqref{eq:universal_eb} with $\EE{\mu_i^2}<\infty,\; \EE{Z_{ij}^2} < \infty$, it holds that:
\begin{equation}
\label{eq:threebenchmarks}
\bayesRisk{\B}{G,F} \leq \AvgBayesRisk{\B-1}{G, F} \leq \bayesRisk{\B-1}{G,F} \, - \, \EE{\VarInline{   m^*(\ordX_i) \mid \mu_i, \alpha_i}}\big/ \B.
\end{equation}
\end{prop}

\begin{rema} 
\label{rema:rates_of_diffs}
For sufficiently ``regular'' problems, $\bayesRisk{\B}{G, F}$ will typically be of order $O(1/\B)$, while \smash{$\bayesRisk{\B-1}{G, F} - \bayesRisk{\B}{G, F}$} will be of order $O(1/\B^2)$; we make this argument rigorous 
for location families in Section~\ref{subsec:loc_family}
and \supplementname~\ref{sec:loc_family_example}. The correction term on the right hand side of~\eqref{eq:threebenchmarks} will also be of order $O(1/\B^2)$ in such problems. In our next example, we demonstrate the importance of this additional term.
\end{rema}

\begin{exam}[Normal likelihood with Normal prior]
\label{exam:normal_normal_regret} We consider the 
Normal-Normal model \eqref{eq:normal_model_unknown_variance} from the introduction with prior variance $A$ and noise variance $\sigma^2=1$, and with replicates $Z_{ij}, j=1, \dotsc, \B$ as in~\eqref{eq:universal_eb}. In this case,\footnote{The details are given in \supplementname~\ref{sec:example_details}.} we can analytically compute that $\bayesRisk{\B}{G, F} = A/(KA+1)$, and so it follows that $\bayesRisk{\B-1}{G, F} - \bayesRisk{\B}{G, F} = \Theta(1/\B^2)$. The right-most inequality in~\eqref{eq:threebenchmarks} is an equality and \smash{$\AvgBayesRisk{\B-1}{G, F}  - \bayesRisk{\B}{G, F} = \Theta(1/\B^4)$}. We conclude that, in the Normal-Normal model, the averaged oracle $\AvgBayesRule$ has risk closer to the full Bayes estimator than to the Bayes estimator based on $\B-1$ replicates.
\end{exam}

\subsection{Regret bound for Aurora}
\label{subsec:regret_aurora}

In this section we derive our main regret bound for Aurora. The basic idea is that the in-sample prediction error of the regression method $\hat{m}$ directly translates to bounds on the estimation error for the effect sizes $\mu_i$, and thus, Aurora can leverage black-box predictive models.  To this end, we define the in-sample prediction error for estimating the averaged oracle $\AvgBayesRule$,
\begin{equation}
\label{eq:in_sample_error_avg}
\overline{\Err}\p{m^*, \hat{m}} := \frac{1}{\n}\sum_{i=1}^\n \EE[G, F]{\cb{\AvgBayesRule(\boldZ_i) - \frac{1}{\B}\sum_{j=1}^{\B} \hat{m}_j\p{\ordX_i(j)}}^2}.
\end{equation}
When the predictive mechanism of the Aurora algorithm (Table~\ref{tab:algorithm}) is the same for all $j$,\footnote{That is, the map \smash{$(\ordX_i(j), Y_i(j))_{i=1,\dotsc,\n} \mapsto \hat{m}_j$} is the same for all $j$. This does not imply that $\hat{m}_j = \hat{m}_{j'}$ for $j \neq j'$, since the training data changes.} then Jensen's inequality provides the following convenient upper bound on~\eqref{eq:in_sample_error_avg}:
\begin{equation}
\label{eq:in_sample_error}
\overline{\Err}\p{m^*, \hat{m}} \leq \, \Err\p{m^*, \hat{m}} := \frac{1}{\n}\sum_{i=1}^\n \EE[G, F]{\cb{m^*\p{\ordX_i} - \hat{m}\p{\ordX_i}}^2}.
\end{equation}
$\Err\p{m^*, \hat{m}}$ is the in-sample error from approximating $m^*$ by $\hat m$. We are ready to state our main result:

\begin{theo}
\label{theo:eb_decomp}
The mean squared error of the Aurora estimator $\muebds_i$  (described in Table~\ref{tab:algorithm})
satisfies the following regret bound under model~\eqref{eq:universal_eb} with $\EE{\mu_i^2}<\infty,\; \EE{Z_{ij}^2} < \infty$.
$$
\begin{aligned}
\frac{1}{\n}\sum_{i=1}^{\n} \EE[G, F]{\p{\mu_i - \muebds_i}^2} &\leq  \bayesRisk{\B}{G, F}\quad\quad &\text{(Irreducible Bayes error)}\\
&+ 2\p{\AvgBayesRisk{\B-1}{G, F} - \bayesRisk{\B}{G, F}} &\text{(Error due to data splitting)}\\
&+ 2\overline{\Err}\p{m^*, \hat{m}} & \text{(Estimation error)}
\end{aligned}
$$
$\muebds_{i,j}$ (i.e., the Aurora estimator based on a single held-out response replicate)  satisfies the above regret bound with \smash{$\AvgBayesRisk{\B-1}{G, F},\, \overline{\Err}\p{m^*, \hat{m}}$} replaced by  \smash{$\bayesRisk{\B-1}{G,F},\, \Err\p{m^*, \hat{m}}$}.
\end{theo}
\noindent 
As elaborated in Remark~\ref{rema:rates_of_diffs}, the second error term in the above decomposition will typically be negligible compared to $\bayesRisk{\B}{G, F}$; this is
the price we pay for making no 
assumptions about $F$ and $G$. Hence, beyond the
irreducible Bayes error, 
the main source of error depends on how well we can 
estimate $m^*(\cdot)$. Crucially, this error is the in-sample estimation error of $\hat{m}$, which is often easier to analyze and smaller in magnitude than out-of-sample estimation error~\citep{hastie2008elements, chatterjee2013assumptionless, rosset2018fixed}.

\subsection{Aurora with $k$-Nearest-Neighbors is universally consistent}
Theorem~\ref{theo:eb_decomp} demonstrated that a regression model with
small in-sample error translates, through Aurora, to mean estimates with small mean squared error. Now, we combine this result with results from nonparametric regression to prove that under model~\eqref{eq:universal_eb},  it is possible to asymptotically (in $\n$) match the oracle risk~\eqref{eq:avgbayesrule} and in view of Proposition~\ref{prop:averaged_jackknife}, to outperform the Bayes risk based on $\B-1$ replicates.

\begin{theo}[Universal consistency with $k$-Nearest-Neighbor ($k$NN) estimator]
\label{theo:johns}
Consider model~\eqref{eq:universal_eb} with \smash{$\EEInline{\mu_i^2}<\infty,\;\EEInline{Z_{ij}^2} < \infty$}. We estimate $\mu_i$ with the Aurora algorithm where $\hat{m}(\cdot)$ is the $k$-Nearest-Neighbor ($k$NN) estimator with $k=k_{\n} \in \mathbb{N}$, i.e., the nonparametric regression estimator which predicts\footnote{The definition below assumes no ties. In the proof we explain how to randomize to deal with ties.}
$$ \hat{m}(\boldx) = \frac{1}{k} \sum_{i \in \mathcal{S}_k(\boldx)} Y_i, \;\; \text{where } \small{ \mathcal{S}_{k}(\boldx) = \cb{i \in \cb{1,\dotsc,\n}: \sum_{j \neq i} \ind\p{ \Norm{\ordX_i - \boldx}_2 > \Norm{\ordX_j - \boldx}_2} < k }},$$
and $\Norm{\cdot}_2$ is the Euclidean distance. If $k=k_{\n}$ satisfies $k \to \infty$, $k/\n \to 0$ as $\n \to \infty$, then:
$$ \limsup_{\n \to \infty}\frac{1}{\n}\sum_{i=1}^{\n}\EE{\p{\mu_i - \muebds_i}^2} =  \AvgBayesRisk{\B-1}{G, F} \leq \bayesRisk{\B-1}{G,F}.$$
\end{theo}

\noindent This result is a consequence of universal consistency in nonparametric regression~\citep*{stone1977consistent, gyorfi2006distribution}. It demonstrates that Aurora can asymptotically match the Bayes risk with substantial generality,\footnote{ Vernon~\citet{johns1957non} established a result similar to Theorem~\ref{theo:johns}. We find this remarkable, since \citet{johns1957non} anticipated later developments in nonparametric regression, independently proving universal consistency for partition-based regression estimators as an intermediate step.} and suggests the power and expressivity of the Aurora algorithm.

To apply Aurora with $k$NN, a data-driven choice of $k$ is required. In \supplementname~\ref{subsec:auroraknn_crossval} we describe a procedure that chooses the number of nearest neighbors $k_j$ per held-out response replicate $j$ through leave-one-out (LOO) cross-validation, and we study its performance empirically in the simulations of Section~\ref{sec:empirical_simulations}. Nevertheless, Aurora-$k$NN tuned by LOO, is computationally involved. This motivates our next section; there we study a procedure that is interpretable, easy to implement and that scales well to large datasets.

\section{Aurora with Linear Regression}
\label{sec:linear_ebds}
In this section we analyze the Aurora algorithm when linear regression is used as the predictive model. That is, $\hat m$ is a linear function of the 
order statistics 
\begin{equation}
\hat m(\ordX_i) =  \hat{\beta}_0 + \sum_{j'=1}^{\B-1} \hat{\beta}_{j'} X_i^{(j')},
\label{eq:linear_predictor}
\end{equation}
where $\hat{\beta}$ are the ordinary least squares coefficients of the linear regression $Y \sim \ordX$. We call the method Auroral, with the final ``l'' signifying ``linear'' and write $\muebdsl_i$ for the resulting estimates of $\mu_i$. Our main result is that Auroral matches the performance of the best estimator that is linear in the order statistics based on $\B-1$ replicates.  To state this result formally, we first define the minimum 
risk among estimators in class $\Ccal$:
\begin{equation}
\label{eq:bayes_risk_ccal}
\Rcal^{\Ccal}_{\B-1}(G, F) := \inf_{m \in \Ccal} \cb{ \frac{1}{\n}\sum_{i=1}^\n \EE[G, F]{\p{\mu_i - m\p{\ordX_i}}^2}}.
\end{equation}
The class of interest to us specifically 
is the class of estimators linear in the order statistics:
\begin{equation}
\label{eq:linear_class}
\text{Lin} := \text{Lin}\p{\RR^{\B-1}} := \cb{ m: m(x) = \beta_0 + \sum_{j'=1}^{\B-1} \beta_{j'} x^{(j')}}.
\end{equation}
This is a broad class that includes all estimators of $\mu_i$ based on \smash{$\ordX_i$} that proceed through the following two steps: (1) summarization, where an appropriate summary statistic of the likelihood $F$ is used, that is linear in the order statistics, and (2) linear shrinkage, where the 
summary statistic is {\em linearly} shrunk towards a fixed location. Schematically:
\vspace{-0.3cm}
\begin{equation*}
\begin{aligned}
  &\text{Step 1 (Summarization)}:\; &\boldX_i \mapsto T\p{\boldX_i} \in \RR,\; \text{where}\;\; &T(\cdot) \in \left\{\begin{matrix*}[l] \text{Sample mean}
\\ \text{Sample median}
\\ \text{Trimmed mean}
\\ \;\;\;\;\; \vdots
\end{matrix*}\right.\\
    &\text{Step 2 (Linear shrinkage)}: &T\p{\boldX_i} \mapsto \alpha T\p{\boldX_i} + \gamma,\;\;\; &\text{(e.g. James-Stein shrinkage)}
\end{aligned}
\end{equation*}
The summarization step can apply non-linear functions of 
the original data, such as the median and the trimmed mean, 
because the 
input to the linear model are the order statistics, rather 
than the original data. Such linear combinations of the order statistics are known as $L$-statistics, which is a class so large that it includes 
efficient estimators of the mean in any smooth, symmetric location family, cf.~\citet{van2000asymptotic} and Section~\ref{subsec:loc_family}.

We now show that Auroral matches 
$\linearRisk{\B-1}{G, F}$ asymptotically in $\n$.

\begin{theo}[Regret over linear estimators]
\label{theo:lin_regret} 
$\phantom{1}$\\
\begin{enumerate}[(i)]
\item Assume there exists $C_{\n} > 0$ such that $\EE{ \displaystyle\max_{i=1,\dotsc,\n} \VarInline{Y_i \cond \ordX_i}} \leq C_{\n}$,\footnote{
By the 
law of total variance, almost surely,  it holds that
\begin{align}
\Var{Y_i \cond \ordX_i} &= \Var{ \EE{Y_i \cond \mu_i, \alpha_i, \ordX_i} \cond \ordX_i} + \EE{ \Var{Y_i \cond \mu_i, \alpha_i, \ordX_i} \cond \ordX_i } \nonumber \\
&= \hspace{0.95in} \Var{ \mu_i \cond \ordX_i} + \EE{ \sigma_i^2 \cond \ordX_i}.
\label{eq:conditional_variance_formula}
\end{align}
Thus, for example, when 
$\mu_i, \sigma_i^2$ have bounded support, there exists $C>0$ such that $\Var{Y_i \cond \ordX_i} \leq C$ almost surely and so, the stated assumption holds with  $C_{\n} = C$.}
then:
$$ \frac{1}{\n}\sum_{i=1}^{\n} \EE{\p{\mu_i - \muebdsl_i}^2} \leq \linearRisk{\B-1}{G, F} +  C_{\n}\frac{\B}{\n}.$$
\item Assume there exists $\Gamma >0$, such that $\EE{Y_i^4} \leq \Gamma^2$, then:
$$ \frac{1}{\n}\sum_{i=1}^{\n} \EE{\p{\mu_i - \muebdsl_i}^2} \leq \linearRisk{\B-1}{G, F} + \Gamma \sqrt{\frac{\B}{\n}}.$$
\end{enumerate}
\noindent If $m^* \in \linclass$, then the conclusion of Theorem~\ref{theo:eb_decomp} holds for $\muebdsl$ with the term $\overline{\Err}\p{m^*, \hat{m}}$ bounded by $C_{\n} \B /\n$ (under Assumption (i)), resp. $\Gamma \sqrt{\B/\n}$ (under (ii)).
\end{theo}

In datasets, we typically 
encounter $\B \ll \n$. 
Then, Theorem~\ref{theo:lin_regret} implies that Auroral will almost match the risk $\linearRisk{\B-1}{G}$ of the best $L$-statistic based on $\B-1$ replicates. This result is in the spirit of retricted empirical Bayes~\citep{griffin1971optimal, maritz1974aligning, norberg1980empirical, robbins1983some}, which seeks to find the best estimator among estimators in a given class, such as the best Bayes linear estimators of~\citet{hartigan1969linear}. In these works linearity typically refers to linearity in $\boldX_i$ and not \smash{$\ordX_i$};~\citet{lwin1976optimal} however uses empirical Bayes to learn the best L-statistic from the class $\text{Lin}$, when the likelihood takes the form of a known location-scale family.

\subsection{Examples of Auroral estimation}
\label{subsec:ebds_examples}
In this section we give three examples in which Auroral satisfies strong risk guarantees.

\begin{exam}[Point mass prior] Suppose 
the prior on $\mu$ is a point mass at $\bar\mu$; that is, 
\smash{$\PP[G]{ \mu_i = \bar{\mu}}=1$}. Then, the Bayes 
rule based on the order statistics, 
\smash{$m^*(\ordX_i) \equiv \bar\mu$}, has risk 0 and is trivially a 
member of $\linclass$, so $\linearRisk{\B-1}{G, F} = 0$. Therefore, by Theorem~\ref{theo:lin_regret}, provided that $\sigma_i^2 \leq C$ almost surely  ($\sigma_i^2 = \Var{Z_{ij} \mid \mu_i, \alpha_i}$), 
the risk of the Auroral estimator satisfies 
$$\frac{1}{\n} \sum_{i=1}^\n \EE{\p{\mu_i - \muebdsl_i}^2} \leq C \frac{\B}{\n}.$$
\end{exam}

\begin{exam}[Normal likelihood with Normal prior]
\label{exam:normal_normal} We revisit Example~\ref{exam:normal_normal_regret},\footnote{See \supplementname~\ref{sec:example_details} for details regarding the regret bounds we discuss here.} wherein we argued that the averaged oracle $\AvgBayesRule(\boldZ_i)$~\eqref{eq:avgbayesrule}  has risk equal to the Bayes risk $\bayesRisk{\B}{G,F}$ plus $O(1/K^4)$. 
Auroral attains this risk, plus the linear least squares estimation error that decays as $O(K/N)$. 

In addition to Auroral, we consider two estimators, that take advantage of the fact that $\bZ_i := \frac{1}{\B} \sum_{j=1}^\B Z_{ij}$ is sufficient for $\mu_i$ in the 
Normal model with $\B$ replicates:
\begin{itemize}[leftmargin=*,wide]
   \item \citet{coey2019improving} (CC-L): We proceed similarly to the Auroral algorithm with one modification. For the $j$-th held-out replicate we regress $Y_i(j)$ on
    $\bX_i(j) = \frac{1}{\B-1}\sum_{j' \neq j} Z_{ij'}$ using ordinary least squares to obtain \smash{$\hat m_j^{\text{CC-L}}$} and  \smash{$ \hat\mu_{ij}^{\text{CC-L}} := \hat m_j^{\text{CC-L}}(\bX_i(j))$}. Finally we estimate $\mu_i$  by \smash{$\hat\mu_{i}^{\text{CC-L}} := \frac{1}{\B}\sum_{j=1}^{\B}\hat\mu_{ij}^{\text{CC-L}}$}.
    \item \citet{james1961estimation} (JS): We estimate $\mu_i$ by $\displaystyle \hat{\mu}^{\text{JS}}_i= \p{1 - \frac{(\n-2)/\B}{\sum_{i=1}^{\n} {\bZ_i}^2}} \bZ_i$.
\end{itemize} 

CC-L implicitly uses 
the assumption of Normal likelihood by reducing the replicates 
to their mean, which is the sufficient statistic. Similar to Auroral, its risk is equal to \smash{$\bayesRisk{\B}{G,F}$} plus $O(1/K^4)$ and the least squares estimation error, which in this case is $O(1/N)$ (only the slope and intercept need to be estimated for each held-out replicate). James-Stein makes full use of the model assumptions 
(Normal prior, Normal likelihood, known variance) and is 
expected to perform best. JS achieves the Bayes risk based on 
$\B$ observations, \smash{$\bayesRisk{\B}{G,F}$} plus an error term that decays as \smash{$O\p{1/(\n\B^2)}$}. The price Auroral pays compared to CC-L and JS for using no assumptions whatsoever, when reduction to the 
mean was possible, is modest.
\end{exam}

\begin{exam}[Exponential families with conjugate priors]
\label{exam:conjugate_exponential_families}
Let $\nu$ be a $\sigma$-finite measure on the Borel sets of $\RR$ and $\mathfrak{X}$ be the interior of the convex hull of the support of $\nu$. Let \smash{$M(\theta) = \log\p{\int \exp(z \cdot \theta)d\nu(z)}$}. Suppose $\Theta = \cb{\theta \in \RR: M(\theta) < \infty}$ is open, that both $\mathfrak{X}, \Theta$ are non-empty and that $M(\cdot)$ is differentiable for $\theta \in \Theta$ with strictly increasing derivative $\theta \mapsto M'(\theta)$.

Now, suppose $\mu_i, Z_{ij}$ are generated from model~\eqref{eq:universal_eb} as follows (with $G,F$ implicitly defined). First,  $\theta_i$ is drawn from a prior with Lebesgue density proportional to $\exp(\B_0 (z_0 \cdot \theta -  M(\theta)))$ on $\Theta$ for some $\B_0 > 0$ and $z_0 \in \mathfrak{X}$. Next, $\mu_i = M'(\theta_i)$ and $Z_{ij} \cond \mu_i$ is drawn from the distribution with $\nu$-density equal to $\exp\p{z \cdot \theta_i - M(\theta_i)}$, where \smash{$\theta_i = (M')^{-1}(\mu_i)$}. \citet[Theorem 2]{diaconis1979conjugate} then prove that,
$$m^*(\ordX_i) = \EE{\mu_i \cond \ordX_i} = \frac{\B_0 z_0  + (\B-1) \bar X_i}{\B_0 + \B -1}.$$
Thus, $m^* \in \linclass$ and Auroral matches the risk of $m^*$ up to the error term in Theorem~\ref{theo:lin_regret}.
\end{exam}

\subsection{Auroral estimation in location families}
\label{subsec:loc_family}
In this section we provide another example of Auroral estimation. In contrast to the rest of this paper, which treats $\B$ as fixed, we consider an asymptotic regime in which both $\B$ and $\n$ tend to $\infty$ (with $\B$ growing substantially slower than $\n$). In doing so, we hope to provide the following conceptual insights: first, we provide a concrete setting in which Auroral dominates any method that first summarizes $\boldZ_i$ as $\bar{Z}_i$. Second, we elaborate on the expressivity of the class $\linclass$~\eqref{eq:linear_class}.

We consider model~\eqref{eq:universal_eb} with $\mu_i \sim G$ for a smooth prior $G$ and \smash{$Z_{ij} \simiid F(\cdot \mid \mu_i)$}, where $F(\cdot \mid \mu_i)$ has density $f( \cdot \; - \;\mu_i)$ with respect to the Lebesgue measure and $f(\cdot)$ is a density that is symmetric around $0$.

We proceed with a heuristic discussion, that we will make rigorous in the formal statements below. Suppose for now that $f(\cdot)$ is sufficiently regular with Fisher Information,\footnote{$\;\,$In location families, the Fisher information is constant as a function of the location parameter $\mu_i$.}
\begin{equation}
    \label{eq:fisherinfo}
    \mathcal{I}(f) = \int \frac{f'(x)^2}{f(x)}\ind(f(x)>0)dx\, < \infty.
\end{equation}
Then, by classical parametric theory, we expect that $\B \cdot \bayesRisk{\B}{G, F} \to \mathcal{I}(f)^{-1}$ as $\B \to \infty$\footnote{Since $\B \to \infty$, the likelihood swamps the prior. Furthermore, we will place regularity assumptions on the prior to rule out the possibility of superefficiency.}. On the other hand, another classical result in the theory of robust statistics~\citep*{bennett1953asymptotic, jung1956linear, chernoff1967asymptotic, van2000asymptotic} states that for smooth location families, there exists an L-statistic (i.e., a linear combination of the order statistics) that is asymptotically efficient. By Theorem~\ref{theo:lin_regret}, we thus anticipate that the risk of Auroral is equal to $(1+o(1))\mathcal{I}(f)^{-1}/\B$.

In contrast, if we first summarize $\boldZ_i$ by $\bar{Z}_i$, then the best estimator we can possibly use is the posterior mean, $\EE{\mu_i \mid \bar{Z}_i}$.  However, for $\B$ large (so that the likelihood swamps the prior), $\EE{\mu_i \mid \bar{Z}_i} \approx \bar{Z}_i$, and so the risk will behave roughly as $\sigma^2/\B$, where $\sigma^2=\int f^2(x)dx$. We summarize our findings in the Corollary below, and provide the formal proofs in \supplementname~\ref{sec:loc_family_example}.

\begin{coro}[Smooth location families]
\label{coro:smoothloc}
Suppose that $G$ satisfies regularity  Assumption~\ref{assum:regular_prior} and $f(\cdot)$ satisfies regularity Assumption~\ref{assum:regular_loc} (where both assumptions are stated in \supplementname~\ref{subsec:regularity}). Then, in an asymptotic regime with $\B, \n \to \infty, \B^2/\n \to 0$, it holds that,
$$\frac{1}{\n}\sum_{i=1}^{\n} \EE{\p{ \mu_i - \muebdsl_i}^2} \bigg/ \frac{\mathcal{I}(f)^{-1}}{\B} \to 1,\;\; \frac{1}{\n}\sum_{i=1}^{\n} \EE{\p{ \mu_i - \hat{\mu}^{\text{Avg}}_i}^2} \bigg/ \frac{\sigma^2}{\B} \to 1 \text { as } \n \to \infty,$$
where $\hat{\mu}^{\text{Avg}}_i$ can be either the CC-L estimator $\hat{\mu}^{\text{CC-L}}_i$ or the Bayes estimator $\EE{\mu_i \mid \bar{Z}_i}$. 
\end{coro}
Recalling that $\mathcal{I}(f)^{-1} \leq \sigma^2$ with equality when $f(\cdot)$ is the Gaussian density, we see that Auroral adapts\footnote{
This adaptivity is perhaps expected, in light of existing 
theory on semiparametric efficiency in location families. 
For example, it is 
known that even for $\n=1$ one can asymptotically (as $\B\to 
\infty$) match the variance of the parametric maximum likelihood estimator in symmetric location families, even without 
precise knowledge of $F$~\citep{stein1956efficient, bickel1993efficient}.
However, the simulations of Section~\ref{subsec:first_illustration} demonstrate that Aurora adapts to the unknown likelihood already for $\B=10$, 
while semiparametric efficiency results are truly asymptotic in $\B$, requiring an initial nonparametric density estimate. 
} to the unknown density $f(\cdot)$ and outperforms any estimator that first averages $\boldZ_i$.

What about location families that are not regular? Below we give an example, namely the Rectangular location family with $f(\cdot) = U[-B,B]$, $B>0$ in which a similar conclusion to Corollary~\ref{coro:smoothloc} holds. The advantage of Auroral is even more pronounced in this case ($O(K^{-2})$ risk for Auroral, versus $O(K^{-1})$ for any method that first averages $\boldZ_i$).

\begin{coro}[Rectangular location family] 
\label{coro:rectangular_loc}
Suppose $f(\cdot)$ is the uniform density on $[-B,B]$ for a $B>0$, that $G$ satisfies the regularity assumption~\ref{assum:regular_prior} in \supplementname~\ref{subsec:regularity} and that  $K,N\to \infty, \, K^3/N\to 0$. Then,
$$\limsup_{\n \to \infty}\cb{\B \cdot \p{\frac{1}{\n}\sum_{i=1}^{\n} \EE{\p{ \mu_i - \muebdsl_i}^2} \bigg /\frac{1}{\n}\sum_{i=1}^{\n} \EE{\p{ \mu_i - \hat{\mu}^{\text{Avg}}_i}^2}}} \leq 6.$$
\end{coro}

\section{Empirical performance in simulations}
\label{sec:empirical_simulations}

In this section, we study the empirical performance of Aurora and competing empirical Bayes algorithms in three scenarios: homoskedastic location families (Section~\ref{subsec:first_illustration}), heteroskedastic location families (Section~\ref{subsec:heterosk_sims}) and a heavy-tailed likelihood (Section~\ref{subsec:pareto}).

\subsection{Homoskedastic location families} 
\label{subsec:first_illustration}

We start by empirically studying the location family problem from Section~\ref{subsec:loc_family} for $\B=10$ replicates and $\n=10^4$ units. We first generate the means 
$\mu_i$ 
from one of two possible priors $G$, parameterized by a simulation parameter $A$: the normal prior $\mathcal{N}(0.5,A)$ and the three-point discrete prior that assigns equal probabilities to $\cb{-\sqrt{3A/2},\,\, 0,\,\, \sqrt{3A/2}}$. Both prior 
distributions have variance $A$.

We then generate the replicates $Z_{ij}$ around each 
mean $\mu_i$ from one of 
three location families: Normal, Laplace, or 
Rectangular. The parameters of these distributions are 
chosen so that the noise variance is $\sigma^2 = \Var{Z_{ij} \mid \mu_i}= 4$.

We compare eight estimators of $\mu_i$: 
\begin{itemize}[leftmargin=*,wide]
    \item \textbf{Aurora-type methods}: Auroral, Aurora-$k$NN (Aur-$k$NN) with $k$ chosen by leave-one-out cross-validation from the set $\cb{1,\dotsc,1000}$ (as described in \supplementname~\ref{subsec:auroraknn_crossval}) and CC-L (described in example~\ref{exam:normal_normal}).
    \item \textbf{Standard estimators of location:} The mean $\bar{Z}_i$, the median and the midrange (that is, $\hat \mu_i = (\max_j\cb{Z_{ij}}+\min_j\cb{Z_{ij}})/2$).
    \item \textbf{Standard empirical Bayes estimators applied to the averages $\bar{Z}_i$:} James-Stein (positive-part) shrinking towards $\sum_{i=1}^{\n}\bar{Z}_i/\n$ (which we provide with oracle knowledge of  $\sigma^2=4$) and the nonparametric maximum likelihood estimator (NPMLE)  of~\citet{koenker2014convex} (as implemented in the REBayes package~\citep{koenker2017rebayes} in the function ``GLmix''), which is a convex programming formulation of the estimation scheme of~\citet{kiefer1956consistency} and \citet{jiang2009general}. For the NPMLE we estimate the standard deviation for each unit by the sample standard deviation $\hat{\sigma}_i^2$ over its replicates and use the working approximation $\bar{Z}_i \mid \mu_i \stackrel{\cdot}{\sim} \mathcal{N}(\mu_i,\hat{\sigma}_i^2/\B)$.
\end{itemize}

\begin{figure}[!ht]
    \centering
    \includegraphics[width=\linewidth]{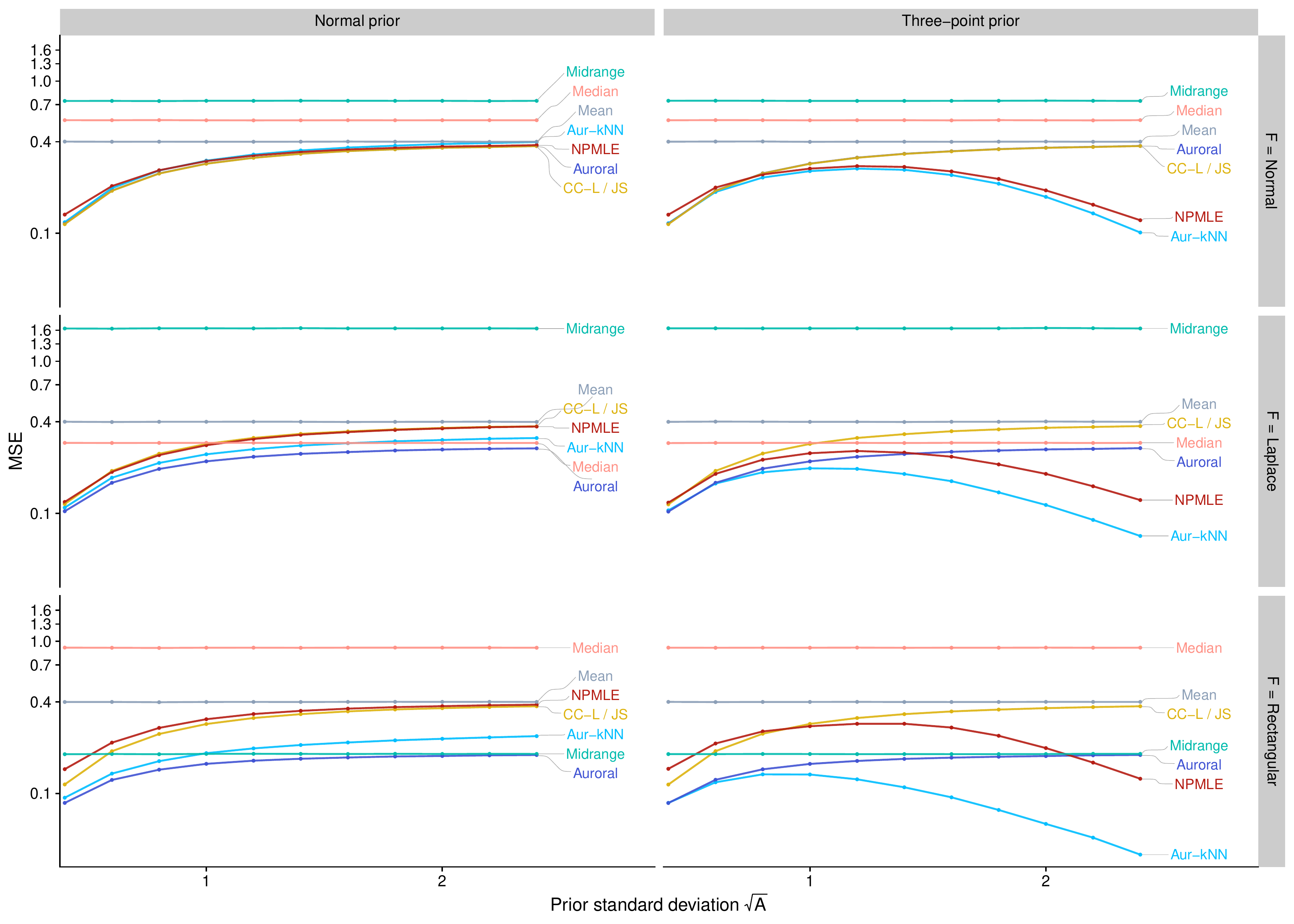}
    \caption{\textbf{Homoskedastic location families:} 
    The MSE of the 8 estimators, as a function of the prior standard deviation. 
    Each column represents a different prior distribution: 
    normal (left) and a three-point distribution (right). 
    Each row represents a different location family 
    for the likelihood (Normal, Laplace, Rectangular).}
    \label{fig:homoskedastic_simulation}
\end{figure}

The results\footnote{Throughout this section we calculate the mean squared error by averaging over 100 Monte Carlo replicates.} are shown in Figure~\ref{fig:homoskedastic_simulation}. The standard location estimators have constant mean squared error (MSE) in all panels, since they do not make use of the prior. 

We discuss the case of a Normal prior first: here the MSE of all methods is non-decreasing in $A$. Auroral closely matches the best estimator for every $A$ and likelihood. In the case of Normal likelihood, James-Stein (with oracle knowledge of $\sigma^2$), CC-L and Auroral perform best,\footnote{CC-L and James-Stein perform so similarly in the simulations of this subsection that they are indistinguishable in all panels of Figure~\ref{fig:homoskedastic_simulation} with the difference of their MSEs smaller than $10^{-3}$ in all cases. In the first panel (Normal prior and likelihood) Auroral is also indistinguishable from CC-L and James-Stein.} followed closely by Aurora-$k$NN and the NPMLE.
The standard location estimators are competitive when the prior is relatively uninformative (i.e., 
$A$ is large). Among these, the mean performs best for the Normal likelihood, the median for the Laplace likelihood and the midrange for the Rectangular likelihood.

The three component prior highlights a case in which non-linear empirical Bayes shrinkage can be helpful. Here, Aurora-$k$NN performs best across all settings, closely followed by the NPMLE in the case of the Normal likelihood\footnote{Even for the Normal likelihood, the assumptions of the implemented NPMLE are not fully satisfied, since we use estimated standard deviations $\hat{\sigma}_i^2$ for each unit.}. The NPMLE is also the second most competitive method for the Laplace likelihood with large prior variance $A$. The behavior of all other methods tracks closely with their behavior under a Normal prior.

\begin{figure}
    \centering
    \includegraphics[width=\linewidth]{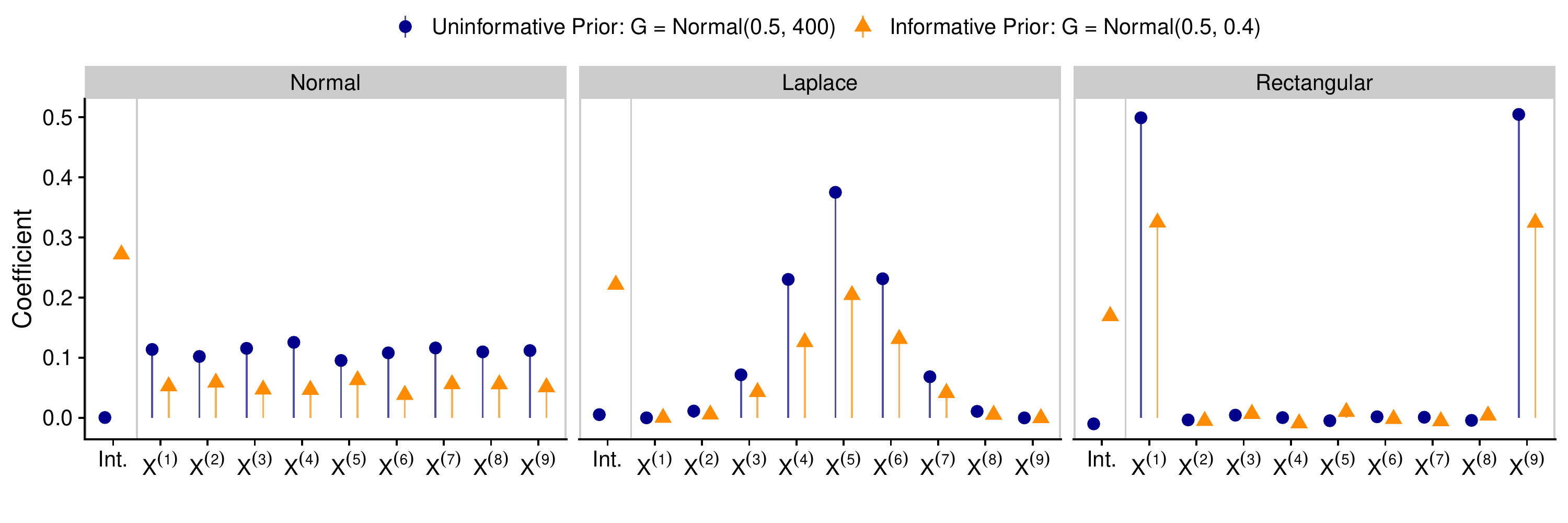}
    \caption{{\small\textbf{The coefficients of the intercept and the order
    statistics in the linear regression model 
    for $\hat m$ in Auroral:} The 
    colors represent different choices of prior $G$ (defined in the legend), while 
    the panels represent different choices of likelihood $F$. The coefficients shown have been averaged over the held-out replicate $j$ of the Auroral algorithm (Table~\ref{tab:algorithm}).}}
    \label{fig:location}
\end{figure}

One may at this point wonder: How do the Auroral weights (coefficients) $\hat\beta$ in 
equation~\eqref{eq:linear_predictor} look like? In Figure~\ref{fig:location}, we show these weights from a single replication (with $\n=2 \cdot 10^5$ and $\B=10$) for each of the three likelihoods considered above and for two Normal priors.
First, we focus on the points in blue, which correspond to 
an uninformative prior ($G=\mathcal{N}(0.5, 400)$). When the likelihood $F$ is Normal, the Auroral weights are roughly constant and 
equal to $1 / 9$. In other words, $\muebdsl_{i,j}$, that is the Auroral fit with held-out replicate $j$, is (approximately) the 
sample mean $\bar X_i(j)$ of $\ordX_i(j)$.
When the likelihood $F$ is Laplace, the Auroral weights pick out the median 
$X_i^{(5)}(j)$ and a few order statistics around it.
When the likelihood $F$ is Rectangular, Auroral assigns approximately
$1/2$ weight each to the minimum and the maximum and $0$ weight 
to all of the other order statistics. In other words, 
$\muebdsl_{i,j}$ is (approximately) the midrange of $\ordX_i(j)$.
Notice that Auroral did not know the likelihood $F$ in any of these examples. Rather, it 
adaptively learned an appropriate summary from the 
data.

Next, we examine the difference between using informative 
versus uninformative priors $G$. When the prior is 
informative ($G=\mathcal{N}(0.5,0.4)$, orange in Figure~\ref{fig:location}), Auroral automatically 
learns a non-zero intercept, which is determined by the 
prior mean, and the remaining weights 
are shrunk towards zero.

\subsection{Heteroskedastic location families}
\label{subsec:heterosk_sims}
In our second simulation setting, we study location families where $\sigma_i^2$ is also random and so we find ourselves in the heteroskedastic location family problem. Again we benchmark Auroral,  Aurora-$k$NN ($k$ chosen by leave-one-out cross-validation from the set $\cb{1,\dotsc,1000}$), CC-L, the Gaussian NPMLE and the sample mean. We also consider two estimators which have been proposed specifically for the heteroskedastic Normal problem, the SURE (Stein's Unbiased Risk Estimate) method of~\citet{xie2012sure} that shrinks towards the grand mean and the GL (Group-linear) estimator of~\citet{weinstein2018group}. We apply these estimators to the averages $\bar{Z}_i$. Both of these estimators have been developed under the assumption that the analyst has exact knowledge of $\sigma_i^2$; so we provide them with this oracle knowledge (SURE (or.) and GL (or.) --- the other methods are not provided this information). Furthermore, we apply the Group-linear method that uses the sample variance $\hat{\sigma}_i^2$ calculated based on the replicates. 

We use three simulations, inspired by simulation settings a), c) and f) of~\citet{weinstein2018group}: in all three simulations we let $\n=10000, \B=10$. First we draw \smash{$\bar{\sigma}^2_i \sim U[0.1, \bar{\sigma}^2_{\text{max}}]$}, where $\bar{\sigma}^2_{\text{max}}$ is a simulation parameter that we vary. Then for the first setting we draw \smash{$\mu_i \sim \mathcal{N}(0,0.5)$}, while for the last two settings we let \smash{$\mu_i = \bar{\sigma}_i^2$}.~\citet{weinstein2018group} use the latter as a model of strong mean-variance dependence. The methods that have access to $\bar{\sigma}_i^2$ can in principle predict perfectly (i.e., the Bayes risk is equal to $0$). Finally we draw \smash{$Z_{ij} \mid \mu_i, \sigma_i^2 \sim F(\cdot \mid \mu_i, \sigma_i^2)$} where $\sigma_i^2 = \bar{\sigma}_i^2 \cdot \B$ and $F$ is either the Normal location-scale family (first two settings) or the Rectangular location-scale family (last setting). 

\begin{figure}[!ht]
    \centering
    \includegraphics[width=\linewidth]{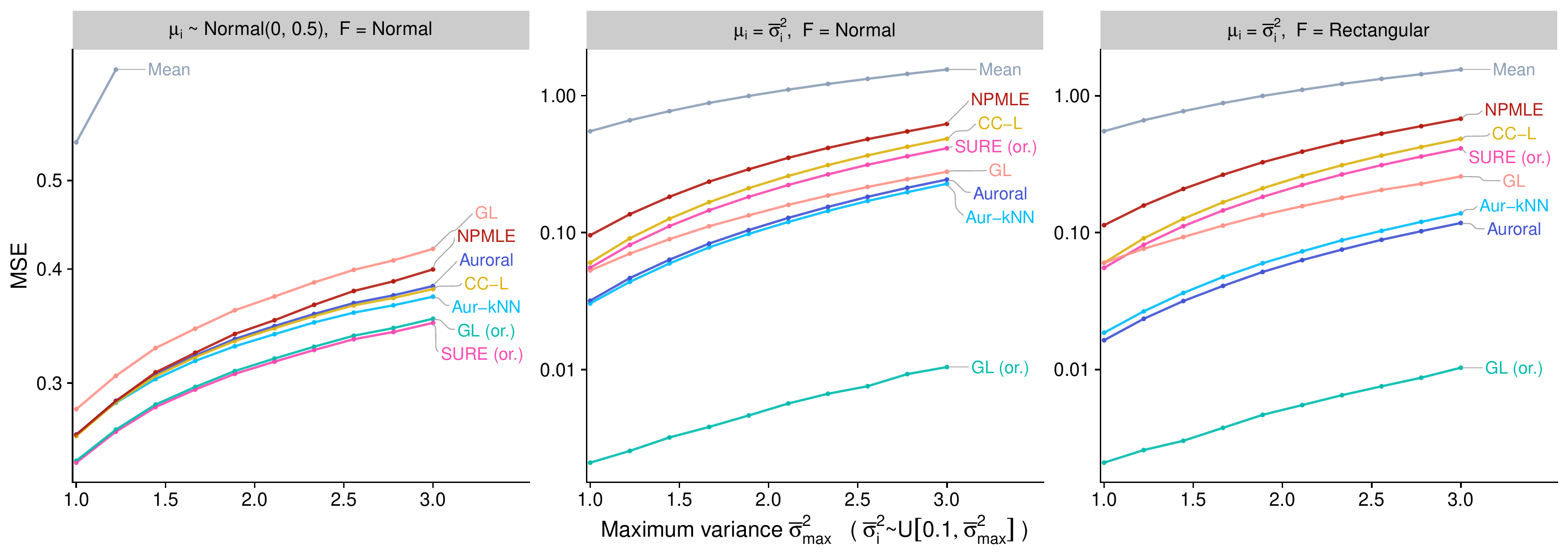}
    \caption{\textbf{Heteroskedastic location families: } Data are generated as follows: First  $\bar{\sigma}^2_i \sim U[0.1, \bar{\sigma}^2_{\text{max}}]$, with $\bar{\sigma}^2_{\text{max}}$ varying on the x-axis. Then $\mu_i \sim \mathcal{N}(0,0.5)$ (in the left panel) or $\mu_i=\bar{\sigma}_i^2$. Finally \smash{$Z_{ij} \mid \mu_i,\sigma_i^2 \sim F(\cdot \mid \mu_i, \sigma_i^2)$}, $j=1,\dotsc,\B$ where $\B=10$, $\sigma_i^2 = \bar{\sigma_i}^2 \B$  and $F(\cdot \mid \mu_i, \sigma_i^2)$ is a Normal location-scale family (first two panels) or Rectangular (last panel). The $y$-axis shows the mean squared error of the estimation methods.}
    \label{fig:heteroskedastic_simulation}
\end{figure}

Results from the simulations are shown in Figure~\ref{fig:heteroskedastic_simulation}. The oracle SURE and oracle Group-linear estimators perform best in the first panel and oracle Group-linear strongly outperforms all other methods in the last two panels. This is not surprising, since oracle Group-linear has oracle access to $\sigma_i^2$ and the method was developed for precisely such settings with strong mean-variance relationship. Among the other methods, Auroral and Aurora-$k$NN remain competitive. In the first panel they match CC-L, while in the last two panels they outperform Group-linear with estimated variances. We point out that Auroral outperforms CC-L in the second panel, despite the Normal likelihood. This is possible, because the mean is no longer sufficient in the heteroskedastic problem.

\subsection{A Pareto example}
\label{subsec:pareto}
For our third example we consider a Pareto likelihood, which is heavy tailed and non-symmetric. Concretely, we let $\mu_i \sim G= U[2, \mu_{\text{max}}]$ (with $\mu_{\text{max}}$ a varying simulation parameter) and $F(\cdot \mid \mu_i)$ is the Pareto distribution with tail index $\alpha=3$ and mean $\mu_i$. We compare Auroral, Aurora-$k$NN (Aur-$k$NN) with $k$ chosen by leave-one-out cross-validation from the set $\cb{1,\dotsc,100}$ (as described in \supplementname~\ref{subsec:auroraknn_crossval}), CC-L, the sample mean and median, as well as the maximum likelihood estimator for the Pareto distribution (assuming the tail index is unknown). For this example we also vary $(\B,\n) = (20, 10^4), (100,10^4), (100,10^5)$. The results are shown in Figure~\ref{fig:pareto_simulation}. Throughout all settings, Auroral performs best, followed by Aurora-$k$NN. All methods improve as $\B$ increases. Auroral and Aurora-$k$NN also improve as $\n$ increases. 

\begin{figure}
    \centering
    \includegraphics[width=\linewidth]{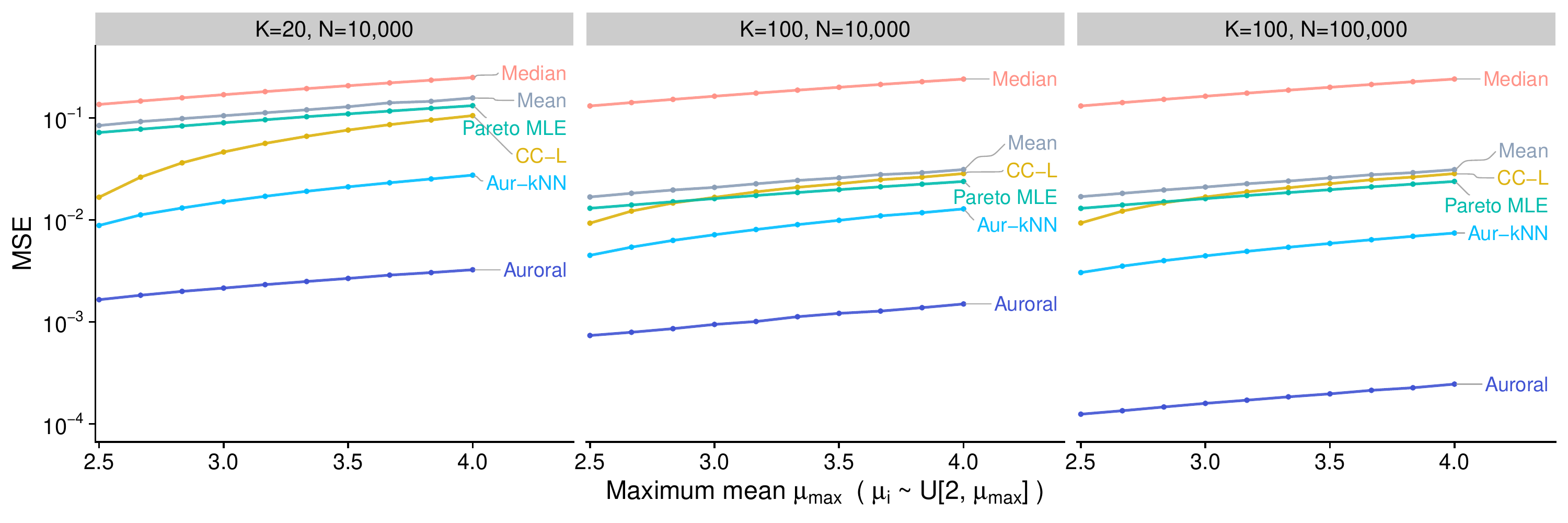}
    \caption{\textbf{Pareto distribution example: } Data are generated as $\mu_i \sim U[2, \mu_{\text{max}}]$, with $\mu_{\text{max}}$ varying on the x-axis and $Z_{ij} \mid \mu_i \sim F(\cdot \mid \mu_i)$, where $F(\cdot \mid \mu_i)$ is the Pareto distribution with mean $\mu_i$ and tail index $\alpha=3$. The panels correspond to different choices for $\B$ and $\n$. The $y$-axis shows the mean squared error of the estimation methods.}
    \label{fig:pareto_simulation}
\end{figure}

\section{Application: Predicting treatment effects at Google}
\label{subsec:cpc}

In this section, we apply Auroral to a problem encountered at Google and other technology companies ---estimating treatment effects at a fine-grained level. All major technology firms run large randomized controlled experiments, often called A/B tests, to study interventions and to evaluate policies~\citep{tang2010overlapping, kohavi2013online,kohavi2017online, athey2019economists}. Estimation of the average treatment effect from such an experiment (e.g., comparing a metric between treated users and control users) is a well-understood statistical task~\citep{wager2016high, athey2017econometrics}. In the application we consider below, instead, interest lies in estimating treatment effects on fine-grained groups -- a separate treatment effect for each of thousands of different online advertisers. In this setting, Empirical Bayes techniques, such as Auroral, can stabilize estimates by sharing information across advertisers.

To apply Auroral for the task of fine-grained estimation of treatment effects, we require replicates. Interestingly, the data from the technology firm we work with is routinely organized and analyzed using `streaming buckets', i.e., the experiment data is divided into $\B$ chunks of approximately equal size that are called streaming buckets. The buckets correspond to (approximately) disjoint subsets of users, partitioned at random, and so data across different buckets are independent to sufficient approximation. We refer to~\citet{chamandy2012estimating} for a detailed description of the motivation for using streaming buckets to deal with the data's scale and structure, as well as statistical and computational issues involved in the analysis of streaming bucket data. For the purpose of our application, each streaming bucket corresponds to a replicate in~\eqref{eq:universal_eb}.

\begin{table}
\centering
\begin{tabular}{l|c c}
&  \multicolumn{2}{c}{$\%\text{ change in }\MSE$} \\
Estimator  & $\Delta\text{CPC}$ &  $\Delta\text{CTR}$\\
\hline 
Aggregate    & Baseline & Baseline \\
Mean &  $-8.1\%\;\phantom{0}(\pm 9.8\%)$ & $-1.0\%\phantom{0}\;(\pm 1.2\%)$ \\
CC-L with Aggregate        & $-49.7\%\;(\pm 5.8\%)$ & $-34.0\% \;(\pm 0.9\%)$\\
CC-L with Mean & $-48.1\%\;(\pm 6.0\%)$ & $-33.7\% \;(\pm 0.9\%)$\\
Auroral & $\mathbf{-63.7\%}\; (\pm 4.2\%)$ & $\mathbf{-37.1\%}\;(\pm 0.9\%)$\\
\end{tabular}
\caption{\textbf{Empirical performance on the advertiser-level estimation problem:} Percent change in mean squared error for estimating change in cost-per-click ($\Delta\text{CPC}$) and change in click-through rate ($\Delta\text{CTR}$) compared to the aggregate estimate ($\pm$ standard errors).}
\label{tab:cpc}
\end{table}

We model the statistical problem of our application as follows. The metric of interest is the cost-per-click (CPC), which is the price a specific advertiser pays for an ad clicked by a user. The goal of the experiment is to estimate the change in CPC before and after treatment. We have data for each advertiser ($\n$ advertisers) and two treatment arms; control ($w=0$) and treated ($w=1$). The data is further divided into $\B$ buckets. For each advertiser $i=1,\dotsc,\n$, bucket $j=1,\dotsc,\B$, and treatment arm $w = 0, 1$, we record the total number of clicks $N_{ijw} \in \mathbb{N}_{>0}$ and the total cost of the clicks $A_{ijw}$. We define $\text{CPC}_{ijw} := A_{ijw}/N_{ijw}$, the empirical cost-per-click for advertiser $i$ in bucket $j$ and treatment arm $w$.
\color{black}
In this application, the advertisers are the units and the buckets are the replicates. Each observation is $Z_{ij} =\text{CPC}_{ij1} - \text{CPC}_{ij0}$, and the goal is to estimate the treatment effect 
\begin{equation}
    \mu_i := \EE{\text{CPC}_{ij1} - \text{CPC}_{ij0} \cond \mu_i, \alpha_i}.
\end{equation}
$\mu_i, \alpha_i$ from model~\eqref{eq:universal_eb} capture advertiser-level idiosyncrasies. We consider the following estimation strategies:

\begin{enumerate}[wide=\parindent]
    \item The aggregate estimate: we pool the data in all buckets and then compute the difference in CPCs, i.e., $\hat{\mu}_i = \sum_{j=1}^{\B} A_{ij1} \big / \sum_{j=1}^{\B} N_{ij1} - \sum_{j=1}^{\B} A_{ij0}\big / \sum_{j=1}^{\B} N_{ij0}$. 
    \item The mean of the $Z_{ij}$, i.e., $\hat{\mu}_i = \frac{1}{\B}\sum_{j=1}^{\B} Z_{ij}$.
    \item The CC-L estimator, wherein for each $j$ we use ordinary least squares to regress $Y_i(j)$ onto the aggregate estimate in the other  $\B-1$ buckets, i.e., onto $ \sum_{j' \neq j} A_{ij'1} \big / \sum_{j'\neq j} N_{ij'1} - \sum_{j' \neq j} A_{ij'0}\big / \sum_{j'\neq j} N_{ij'0}$.
    \item The CC-L estimator, wherein for each $j$ we use ordinary least squares to regress $Y_i(j)$ onto the mean of the other $\B-1$ buckets, i.e., onto  $\frac{1}{\B-1}\sum_{j'\neq j} Z_{ij'}$
    \item The Auroral estimator that (for each $j$) regresses $Y_i(j)$ on the order statistics  $\ordX_i(j)$.
\end{enumerate}

We empirically evaluate the methods as follows: we use data from an experiment running at Google for one week, retaining only the top advertisers based on number of clicks, resulting in $\n >50,000$. The number of replicates is equal to $\B=4$. As ground truth, we use the aggregate estimate based on experiment data from the 3 preceding and 3 succeeding weeks. Then, we compute the mean squared error of the estimates (calculated from the one week) against the ground truth and report the percent change compared to the aggregate estimate.

The results are shown in Table~\ref{tab:cpc}. The table also shows the results of the same analysis applied to a second metric, the change in click-through rate (CTR),  which is the proportion of times that an ad by a given advertiser which is shown to a user is actually clicked. We observe that the improvement in estimation error through Auroral is substantial. Furthermore, Auroral outperforms both variants of CC-L, which in turn outperform estimators that do not share information across advertisers (i.e., the aggregate estimate and the mean of the $Z_{ij}$).

\section{Conclusion}
\label{sec:conclusion}

We have presented a general framework for constructing 
empirical Bayes estimators from $\B$ noisy replicates of $\n$ units. The basic idea of our method, which we term 
Aurora, is to leave one replicate out and 
regress this held-out replicate on the remaining 
$\B -1$ replicates.  We then repeat this process over all choices of held-out replicate and average the results. We have shown that if the $\B -1$ 
replicates are first 
sorted, then even linear regression produces 
results that are competitive with the best 
methods, which usually make parametric assumptions, 
while our method is fully nonparametric.

We conclude by mentioning some direct extensions of Aurora
that are suggested by its connection to regression.

\paragraph{More powerful regression methods:} In this paper, 
we have used linear regression and $k$-Nearest Neighbor regression 
to learn \smash{$\hat m(\cdot)$}.  But we can go 
further;
for example, we could use isotonic regression~\citep{guntuboyina2018nonparametric} based on a partial order on \smash{$\ordX_i$}.
Or we could combine linear and isotonic 
regression by considering single index models with 
non-decreasing link function~\citep{balabdaoui2016least}, i.e., predictors of the form $\hat m(x^{(\cdot)}) = t(\alpha^\top x^{(\cdot)})$, where $\Norm{\alpha}_2=1$ and $t$ is an unknown non-decreasing function. Other possibilities include recursive partitioning \citep{breiman1984classification, zeileis2008model}, in which linear regression is fit on the leaves of a tree, or even random forests aggregated from such trees~\citep{friedberg2018local, kunzel2019linear}.

\paragraph{More general targets:} We have only
considered estimation of \smash{$\mu_i = \EEInline{Z_{ij} \cond \mu_i, \alpha_i}$} in~\eqref{eq:universal_eb}. As pointed out in \citet{johns1957non, johns1986fully} this naturally extends to parameters \smash{$\theta_i = \EEInline{h(Z_{ij})\cond \mu_i, \alpha_i}$} where $h$ is a known function. 
The only modification needed to estimate $\theta_i$ is that 
we fit a regression model to learn \smash{$\EEInline{ h(Y_i) \cond \ordX_i}$} instead. We may further extend Vernon Johns' observation to arbitrary U-statistics. Concretely, given \smash{$r < \B, r \in \mathbb N$} and a fixed function \smash{$h: \RR^r \to \RR$}, we can use Aurora to estimate \smash{$\theta_i = \EE{h(Z_{i1}, Z_{i2}, \dotsc, Z_{ir}) \cond \mu_i, \alpha_i}$}. In this case we need to hold out $r$ replicates to form the response. For example, with $r=2$ and \smash{$h(z_1,z_2)=(z_1-z_2)^2/2$}, we can estimate the conditional variance 
\smash{$\sigma^2_i = \Var{Z_{ij}\cond \mu_i, \alpha_i}$}. Denoising the variance with empirical Bayes is an important problem that proved to be essential for the analysis of genomic data~\citep{smyth2004linear,lu2016variance}. However, these 
papers assumed a parametric form of the likelihood, 
while Aurora would permit fully nonparametric estimation of 
the variance parameter.

\paragraph{External covariates:} Model~\eqref{eq:universal_eb} posits {\em a priori} exchangeability of the $\n$ units. However, in many applications, domain experts also have access to side-information  $\zeta_i$ about each unit. Hence, multiple authors \citep{fay1979estimates, tan2016steinized, kou2017optimal, banerjee2018adaptive, coey2019improving, ignatiadis2019covariate} have developed methods that improve mean estimation by utilizing information in the $\zeta_i$. Aurora can be directly extended to accommodate such side information. Instead of regressing 
$Y_i$ on \smash{$\ordX_i$}, one regresses $Y_i$ on 
both \smash{$\ordX_i$} and $\zeta_i$.

\vspace{1cm}
\subsection*{Software}
We provide reproducible code for our simulation results in the following Github repository:
\url{https://github.com/nignatiadis/AuroraPaper}.
A package implementing the method is available at \url{https://github.com/nignatiadis/Aurora.jl}. The package has been implemented in the Julia programming language~\citep{bezanson2017julia}.

\subsection*{Acknowledgments}
We are grateful to Niall Cardin, Michael Sklar and Stefan Wager 
for helpful discussions and comments on the manuscript. We would also like to thank the Associate Editor and the anonymous reviewers for their insightful and helpful suggestions.

\bibliographystyle{abbrvnat}
\bibliography{eb_datasplit}

\newpage

\begin{appendix}
\setcounter{page}{1}
\renewcommand{\thepage}{S\arabic{page}} 
\setcounter{table}{0}
\renewcommand{\thetable}{S\arabic{table}}
\setcounter{equation}{0}
\renewcommand{\theequation}{S\arabic{equation}}
\setcounter{rema}{0}
\renewcommand{\therema}{S\arabic{rema}}
\setcounter{prop}{0}
\renewcommand{\theprop}{S\arabic{prop}}

\section{Proofs for Section~\ref{sec:method}}
\subsection{Proof for Proposition~\ref{prop:averaged_jackknife}}
\label{subsec:proof_jackknife}

\begin{proof}
For the LHS inequality, it suffices to note that $\AvgBayesRule(\boldZ_i)$ is a function of $\boldZ_i$, and so is dominated by the Bayes rule in terms of mean squared error. For the RHS, we proceed as follows. First, 
$$
\begin{aligned}
\bayesRisk{\B-1}{G,F} &= \EE{ \p{\mu_i - m^*(\ordX_i)}^2} \\
&= \EE{\Var{\mu_i \mid \ordX_i}} \\
&\stackrel{(i)}{=} \Var{\mu_i} - \Var{\EE{\mu_i \mid \ordX_i}} \\
&= \Var{\mu_i} - \VarInline{m^*(\ordX_i)}.
\end{aligned}
$$
$(i)$ follows from the law of total variance and the other equalities are a consequence of the definition  \smash{$m^*(\ordX_i) = \EEInline{\mu_i \mid \ordX_i}$}. Next,
$$
\begin{aligned}
\AvgBayesRisk{\B-1}{G, F} &= \EE{ \p{\mu_i - \AvgBayesRule(\boldZ_i)}^2} \\
&= \Var{ \mu_i - \AvgBayesRule(\boldZ_i)} \\
& = \Var{\mu_i} + \Var{ \AvgBayesRule(\boldZ_i)} -2 \Cov{\mu_i,\, \AvgBayesRule(\boldZ_i)}
\end{aligned}
$$
By linearity,
$$
\begin{aligned}
\Cov{\mu_i, \AvgBayesRule(\boldZ_i)} &= \Cov{\mu_i,\, \frac{1}{\B}\sum_{j=1}^{\B} m^*\p{\ordX_i(j)}} \\
&= \frac{1}{\B}\sum_{j=1}^{\B}\Cov{\mu_i,\, m^*\p{\ordX_i(j)}} \\
&= \Cov{\mu_i,\, m^*(\ordX_i)} \\
&= \Var{m^*(\ordX_i)}.
\end{aligned}
$$
In the last line we used again the fact that \smash{$m^*(\ordX_i) = \EEInline{\mu_i \mid \ordX_i}$}.
We now proceed with the key step of our argument:
$$
\begin{aligned}
\Var{ \AvgBayesRule(\boldZ_i)} &\stackrel{(i)}{=} \EE{ \Var{\AvgBayesRule(\boldZ_i) \mid \mu_i, \alpha_i}} + \Var{\EE{\AvgBayesRule(\boldZ_i) \mid \mu_i, \alpha_i}} \\ 
&\stackrel{(ii)}{\leq} \frac{\B-1}{\B}\EE{\Var{m^*(\ordX_i) \mid \mu_i, \alpha_i}} + \Var{\EE{m^*(\ordX_i) \mid \mu_i, \alpha_i}} \\
&\stackrel{(iii)}{=}  \Var{m^*(\ordX_i)} -  \frac{1}{\B}\EE{\Var{m^*(\ordX_i) \mid \mu_i, \alpha_i}}.
\end{aligned}
$$
$(i)$ and $(iii)$ follow from the law of total variance. For $(ii)$ we used two results: for the right part, we used the fact that $\EE{\AvgBayesRule(\boldZ_i) \mid \mu_i, \alpha_i}=\EEInline{m^*(\ordX_i) \mid \mu_i, \alpha_i}$. For the left part, as announced in the main text, we used Theorem 2 of~\citet{efron1981jackknife}.

We conclude by combining the preceding displays.
\end{proof}
\subsection{Proof for Theorem~\ref{theo:eb_decomp}}
\begin{proof}
We first obtain a decomposition for a single coordinate $i$.
For simplicity, we suppress the dependence on $G$ and $F$  and first prove the result for $\hat\mu_i := \muebds_{i,j}$, i..e, Aurora based on a single held-out response replicate $j \in\cb{1,\dotsc,\B}$.
\begin{align}
    \EE{\p{\mu_i - \hat{\mu}_i}^2} &= \EE{\p{\mu_i - \EE{\mu_i \cond \boldZ} + \EE{\mu_i \cond \boldZ} - \hat{\mu}_i}^2} \nonumber\\
    &= \EE{\p{\mu_i - \EE{\mu_i \cond \boldZ}}^2} +  \EE{\p{\EE{\mu_i \cond \boldZ} - \hat{\mu}_i}^2} \nonumber\\
    &= \EE{\p{\mu_i - \EE{\mu_i \cond \boldZ_i}}^2} +  \EE{\p{\EE{\mu_i \cond \boldZ_i} - \hat{\mu}_i}^2} \nonumber\\
    &= \bayesRisk{\B}{G, F} + \EE{\p{\EE{\mu_i \cond \boldZ_i} - \hat{\mu}_i}^2} \label{eq:standard_bayes_mse_decomp}
\end{align}
To see why the cross-term $\EE{\left(\mu_i - \EE{\mu_i \cond \boldZ} \right) \left( \EE{\mu_i \cond \boldZ} - \hat{\mu}_i \right)}$ vanishes in the second equality above, 
observe that the factor 
$\left(\EE{\mu_i \cond \boldZ} - \hat{\mu}_i\right)$ is 
measurable with respect to $\boldZ$. 
Therefore, the expectation conditional on $\boldZ$ 
is zero (almost surely), so the unconditional expectation 
(i.e., the cross-term) is also zero.

Next, we examine the quantity inside the expectation in 
the second term of \eqref{eq:standard_bayes_mse_decomp}. By adding and subtracting $m^*(\ordX_i) := \EE{\mu_i \cond \ordX_i}$ and using the inequality $(a + b)^2 \leq 2a^2 + 2b^2$, we obtain
\begin{align}
 \p{\EE{\mu_i \cond \boldZ_i} - \hat{\mu}_i}^2 &= \p{\p{\EE{\mu_i \cond \boldZ_i} - m^*(\ordX_i)} + \p{ m^*(\ordX_i) - \hat{\mu}_i}}^2 \nonumber \\
 &\leq 2\p{\EE{\mu_i \cond \boldZ_i} - m^*(\ordX_i)}^2 + 2 \p{ m^*(\ordX_i) - \hat{\mu}_i}^2. \label{eq:second_term_decomp}
\end{align}
Now, we take the expectation of \eqref{eq:second_term_decomp}. 
For the first term, we can repeat 
the argument from \eqref{eq:standard_bayes_mse_decomp}
to obtain 
\begin{align*}
    \bayesRisk{\B-1}{G, F} &= \EE{\p{\mu_i - m^*(\ordX_i)}^2} \\
    &= \EE{\p{\mu_i - \EE{\mu_i \mid \boldZ}}^2} + \EE{\p{\EE{\mu_i \mid \boldZ} - m^*(\ordX_i)}^2} \\
    &=\bayesRisk{\B}{G, F} + \EE{\p{\EE{\mu_i \mid \boldZ} - m^*(\ordX_i)}^2}.
\end{align*}
which can be rearranged to show that 
$$ 2\EE{\p{\EE{\mu_i \cond \boldZ_i} - m^*(\ordX_i)}^2} = 2(\bayesRisk{\B-1}{G, F} - \bayesRisk{\B}{G, F}). $$
To summarize, we have the following result 
for a single coordinate $i$:
\begin{align*}
    \EE{(\mu_i - \hat\mu_i)^2} \leq \bayesRisk{\B}{G, F} + 2(\bayesRisk{\B-1}{G, F} - \bayesRisk{\B}{G, F}) + 2 \EE{\p{ m^*(\ordX_i) - \hat{\mu}_i}^2}.
\end{align*}
Finally, we average over all $i$ to obtain the desired result:
\begin{align*}
    \frac{1}{\n} \sum_{i=1}^{\n} \EE{(\mu_i - \hat\mu_i)^2} &\leq \bayesRisk{\B}{G, F} + 2(\bayesRisk{\B-1}{G, F} - \bayesRisk{\B}{G, F}) + 2 \frac{1}{\n} \sum_{i=1}^\n \EE{\p{ m^*(\ordX_i) - \hat{\mu}_i}^2} \\
    &= \bayesRisk{\B}{G, F} + 2(\bayesRisk{\B-1}{G, F} - \bayesRisk{\B}{G, F}) + 2 \Err\p{m^*, \hat{m}}.
\end{align*}
The proof for $\hat\mu_i := \muebds_{i}$ is the same verbatim, if we replace $m^*(\ordX_i)$ by $\AvgBayesRule(\boldZ_i)$, $\bayesRisk{\B-1}{G,F}$ by $\AvgBayesRisk{\B-1}{G,F}$ and $\Err\p{m^*, \hat{m}}$ by $\overline{\Err}\p{m^*, \hat{m}}$.
\end{proof}

\section{Aurora with Nearest Neighbors}

\subsection{Proof for Theorem~\ref{theo:johns}}

\begin{proof}
Throughout the proof we assume that ties among the \smash{$\ordX_i(j)$} happen with probability $0$. We avoid ties as follows. As in Chapter 6 of~\citet{gyorfi2006distribution}, we use $k$NN to regress $Y_i(j)$ on \smash{$(\ordX_i(j), U_i(j))$}, where $U_i(j)$ are independently drawn from $U[0,\varepsilon]$ for some small, fixed $\varepsilon > 0$. The regression function remains the same, since almost surely \smash{$\EEInline{Y_i(j) \cond \ordX_i(j)} = \EEInline{Y_i(j) \cond \ordX_i(j), U_i(j)}$}. We will however suppress the $U_i(j)$ from our notation and assume ties do not occur for the \smash{$\ordX_i(j)$}; otherwise the proofs go through verbatim by replacing \smash{$\ordX_i(j)$} by \smash{$(\ordX_i(j), U_i(j))$} in all subsequent arguments.

We first claim that:
\begin{equation}
\label{eq:eb_decomp_no_twos}
\begin{aligned}
&\abs{\frac{1}{\n}\sum_{i=1}^{\n} \EE{\p{\mu_i - \muebds_i}^2} -  \AvgBayesRisk{\B-1}{G, F}} \\
\leq \;\; &\Err(m^*, \hat{m})^{1/2}\p{\Err(m^*, \hat{m})^{1/2} +   2\AvgBayesRisk{\B-1}{G, F}^{1/2}}. 
\end{aligned} 
\end{equation}
To see this, first note that:
\begin{equation}
\p{\mu_i - \hat{\mu}_i}^2=\p{\mu_i - \AvgBayesRule(\boldZ_i)}^2 +  \p{ \AvgBayesRule(\boldZ_i) - \hat{\mu}_i}^2 +  2\p{\mu_i - \AvgBayesRule(\boldZ_i)}\p{ \AvgBayesRule(\boldZ_i) - \hat{\mu}_i}.
\label{eq:decomp_knn}
\end{equation}
By the Cauchy–Schwarz inequality:
$$\abs{\EE{\p{\mu_i - \AvgBayesRule(\boldZ_i)}\p{ \AvgBayesRule(\boldZ_i) - \hat{\mu}_i}}} \leq \EE{\p{\mu_i - \AvgBayesRule(\boldZ_i)}^2}^{1/2} \EE{\p{ \AvgBayesRule(\boldZ_i) - \hat{\mu}_i}^2}^{1/2}.$$
By definition it holds that $\EEInline{\p{\mu_i - \AvgBayesRule(\boldZ_i)}^2} = \AvgBayesRisk{\B-1}{G,F}$. By~\eqref{eq:in_sample_error} and permutation equivariance (with respect to permutations of the units) of the $k$NN estimator $\hat{m}(\cdot)$ (when there are no ties) it holds that:
$$\EE{\p{ \AvgBayesRule(\boldZ_i) - \hat{\mu}_i}^2} = \overline{\Err}\p{m^*, \hat{m}} \leq \Err\p{m^*, \hat{m}}.$$
Taking expectations in~\eqref{eq:decomp_knn}, using the above results and rearranging, we conclude with the claim in~\eqref{eq:eb_decomp_no_twos}. Returning to the main proof, in view of~\eqref{eq:eb_decomp_no_twos}, it suffices to show that:
\begin{equation}
    \label{eq:err_to_zero}
    \limsup_{\n \to \infty} \Err(m^*, \hat{m}) = 0.
\end{equation}
By exchangeability of units in~\eqref{eq:universal_eb} and permutation equivariance of the $k$NN estimator $\hat{m}(\cdot)$ without ties, \eqref{eq:err_to_zero} is equivalent to proving that:
$$\limsup_{\n \to \infty} \EE{ \p{ m^*(\ordX_1) - \hat{m}(\ordX_1)}^2} = 0.$$
We prove this by instead considering the $k$NN estimator $\hat{m}_{\n, -1}(\cdot)$ applied to all observations except the first; however still with the same number of nearest neighbors $k=k_{\n}$ (instead of $k_{\n-1}$). Then:
{\small $$ \EE{ \p{ m^*(\ordX_1) - \hat{m}(\ordX_1)}^2} \leq 2\EE{ \p{ m^*(\ordX_1) - \hat{m}_{\n,-1}(\ordX_1)}^2} + 2\EE{ \p{ \hat{m}_{\n,-1}(\ordX_1) - \hat{m}(\ordX_1)}^2}.$$}
The first of these terms converges to $0$ by existing results on universal consistency in nonparametric regression, concretely Theorem 6.1 of~\citet{gyorfi2006distribution}. We need to show that the second term also converges to $0$. Let $j^*$ be the $k$-th NN of $\ordX_1$ among $\cb{\ordX_2, \dotsc, \ordX_{\n}}$. Then
$$\hat{m}(\ordX_1) - \hat{m}_{\n,-1}(\ordX_1)  = \frac{Y_1 - Y_{j^*}}{k}. $$
Consequently:
$$\EE{\p{\hat{m}_{\n,-1}(\ordX_1) - \hat{m}(\ordX_1)}^2} \leq \frac{ 2\EEInline{Y_1^2}}{k^2} + \frac{2\EEInline{Y_{j^*}^2}}{k^2}. $$
The first term goes to $0$, since we assumed that $\EE{Y_1^2} = \EE{Z_{ij}^2} < \infty$ and $k = k_{\n} \to \infty$. We handle the second term as follows:

$$
\begin{aligned}
\EE{Y_{j^*}^2} &= \EE{ \sum_{i=2}^{\n}Y_i^2\ind\p{ \ordX_i  \text{ is the } k\text{-th NN of } \ordX_1 \text{ in } \cb{\ordX_2, \dotsc, \ordX_{\n}}}} \\
&\leq \EE{ \sum_{i=2}^{\n}Y_i^2\ind\p{ \ordX_i  \text{ is among the } k\text{ NNs of } \ordX_1 \text{ in } \cb{\ordX_2, \dotsc, \ordX_{\n}}}} \\
&= \sum_{i=2}^{\n} \EE{ Y_i^2\ind\p{ \ordX_i  \text{ is among the } k\text{ NNs of } \ordX_1 \text{ in } \cb{\ordX_2, \dotsc, \ordX_{\n}}}} \\
&\stackrel{(i)}{=} \sum_{i=2}^{\n} \EE{ Y_1^2\ind\p{ \ordX_1  \text{ is among the } k\text{ NNs of } \ordX_i \text{ in } \cb{\ordX_1, \dotsc, \ordX_{\n}}\setminus\cb{\ordX_i}}}\\
&=  \EE{ Y_1^2 \sum_{i=2}^{\n}\ind\p{ \ordX_1  \text{ is among the } k\text{ NNs of } \ordX_i \text{ in } \cb{\ordX_1, \dotsc, \ordX_{\n}}\setminus\cb{\ordX_i}}}\\
& \stackrel{(ii)}{\leq}   \EE{Y_1^2 \gamma k} = \gamma k \EE{Y_1^2}.
\end{aligned}
$$
We elaborate on two steps: $(i)$ holds by exchangeability of the $(\ordX_i, Y_i),i=1,\dotsc,\n$. $(ii)$ holds for a constant $\gamma < \infty$ that depends only on the dimension $\B$ by Corollary 6.1. of~\citet{gyorfi2006distribution}. To conclude we divide by $k^2$ and the result follows since we assumed that $k \to \infty$ and $\EE{Y_1^2} < \infty$.

\end{proof}
\newpage
\subsection{Tuning of Aurora-$k$NN by cross-validation}
\label{subsec:auroraknn_crossval}

Aurora-$k$NN is a specific instantiation of the general Aurora algorithm (Table~\ref{tab:algorithm}) with a specific choice of Step 3, which takes the form of leave-one-out cross-validated $k$-Nearest Neighbor regression (see e.g.,~\citet{azadkia2019optimal} and references therein). The box below presents the algorithm in detail:

\LinesNotNumbered
\begin{algorithm}
\SetAlgorithmName{Aurora-$k$NN}{auroraknn}{List}
\caption{Aurora with supervised predictions based on $k$-Nearest Neighbor regression ($k$ chosen by leave-one-out cross-validation).}
\Input{$\boldZ_i$, $i=1,\dotsc,\n$: Replicated samples ($K$ replicates per unit). \\
$k_{\text{max}}$: Integer upper bound on number of nearest neighbors to consider.}
 \hrulealg
    \For{$j=1,\dotsc,K$}{
     \nextnr
    Split the replicates for each unit $i$, $\boldZ_i$,
    into $\boldX_i := (Z_{i1},\dotsc, Z_{i(j-1)}, Z_{i(j+1)}, \dotsc, Z_{i\B})$ and $Y_i := Z_{ij}$, as in \eqref{eq:xy_def}.\;
    \nextnr 
    For each $i$ order the values of $\boldX_i$ to obtain $\ordX_i$.\;
    \nextnr 
    Preprocess \smash{$\ordX_i, \; i=1,\dotsc,\n$} to facilitate nearest neighbor searches.\;
    \For{$i=1,\dotsc,\n$}{
    Let $o(i,1),\dotsc,o(i,k_{\text{max}}-1)$ be the indices of the $k_{\text{max}}-1$ nearest neighbors of $\ordX_i$, excluding $i$ and sorted by increasing distance (i.e., $\NormInline{\ordX_{i'} - \ordX_i}_2$ with $i' \in \cb{1,\dotsc,\n}\setminus\cb{i}$ is minimized by $i'=o(i,1)$). \;
    }
    \For{$k=1,\dotsc,k_{\text{max}}-1$}{
       Compute the leave-one-out cross-validation error of nearest neighbor regression with $k$ neighbors, i.e., $\text{LOO}(k) = \frac{1}{\n}\sum_{i=1}^{\n} \big( Y_i - \frac{1}{k}\sum_{i'=1}^k Y_{o(i,i')}\big)^2.$
    }
    Let $k^*_j \in \argmin_{k\in \cb{1,\dotsc,k_{\text{max}}}}\cb{\text{LOO}(k-1)}$, where $\text{LOO}(0) = \frac{1}{\n} \sum_{i=1}^{\n} Y_i^2$.\;
    \nextnr 
    Compute the $k$NN prediction with $k=k^*_j$ for all $i$, i.e.,
    $$\muebds_{i,j} := \frac{1}{k^*_j}\Bigg(Y_i + \sum_{i'=1}^{k^*_j-1} Y_{o(i,i')}\Bigg).$$
 }
  Return $\muebds_i := \frac{1}{\B}\sum_{j=1}^{\B} \muebds_{i,j}$ for all $i$.\;
\end{algorithm}
Two remarks are in order:
\begin{enumerate}
    \item A naive approach to finding all nearest neighbors in step 3 of the algorithm (say, for a fixed held-out response replicate $j$) has computational complexity $O(\n^2 \B)$. By preprocessing all ordered samples (at the beginning of step 3), this computational complexity may be decreased, especially for small $\B$. In our default implementation we preprocess the data using a $k$d-tree as implemented in the NearestNeighbors.jl~\citep{NearestNeighbors_jl} package when $\B \leq 12$,. For $\B > 12$ we use brute-force search of nearest neighbors.
    \item The leave-one-out calculation can be substantially sped up by reusing the computation from step $k$ when computing $\text{LOO}(k+1)$ via the following elementary identity:
    $$\frac{1}{k+1}\sum_{i'=1}^{k+1} Y_{o(i,i')} = \frac{1}{k+1} \p{ k \cdot \frac{1}{k}\sum_{i'=1}^{k} Y_{o(i,i')} \; + \; Y_{o(i,k+1)}}.$$
\end{enumerate}

\section{Proof for Auroral estimator (Theorem~\ref{theo:lin_regret})}
\label{sec:auroral_proof}
\begin{proof}
Throughout this proof we let \smash{$\projXj$} be the orthogonal projection operator onto the linear space spanned by the columns of \smash{$\ordX(j)$} and the ones vector $(1,\dotsc,1)^\top$. We also use vectorized notation, e.g., $\mu=(\mu_1,\dotsc,\mu_{\n})^\top$, $Y(j) = (Y_1(j),\dotsc,Y_{\n}(j))^\top$. With this notation it holds that $$\muebdsl_i = \frac{1}{\B}\sum_{j=1}^{\B}\muebdsl_{i,j},\;\;\; \muebdsl_{\cdot,j} = \projXj Y(j).$$
First, by Jensen's inequality:
$$
\begin{aligned}
\EE{\p{ \mu_i - \muebdsl_{i}}^2} = \EE{\cb{\frac{1}{\B}\sum_{j=1}^{\B}(\mu_i - \muebdsl_{i,j})}^2} \leq \EE{\frac{1}{\B}\sum_{j=1}^{\B}(\mu_i - \muebdsl_{i,j})^2}.
\end{aligned}
$$
Thus it suffices to bound  $\frac{1}{\n}\sum_{i=1}^{\n}\EE{(\mu_i - \muebdsl_{i,j})^2}$ for a fixed $j$. In doing so, we omit $j$ from the notation, e.g., we write $\ordX$ instead of $\ordX(j)$, $\projX$ instead of $\projXj$ and $Y$ instead of $Y(j)$. We also (with slight abuse of notation) write  \smash{$m^* = m^*(\ordX)$}. It then holds that:

\begin{equation*}
    \begin{aligned}
    &\frac{1}{\n}\sum_{i=1}^{\n} \EE{\p{\mu_i - \muebdsl_{i,j}}^2} \\ 
    =\;\; & \frac{1}{\n} \EE{ \Norm{ \mu - \muebdsl_{\cdot,j}}^2_2} \\
    =\;\; & \frac{1}{\n}\EE{\Norm{ \mu - \projX Y}^2_2} \\
    =\;\; &\frac{1}{\n}\EE{\Norm{\mu - \projX m^* + \projX m^* - \projX Y}^2_2}\\
    =\;\; &\frac{1}{\n}\cb{\EE{\Norm{\mu - \projX m^*}_2^2} +  \EE{\Norm{\projX m^* -\projX Y}_2^2} + 2\EE{ \p{\mu - \projX m^*}^\top\p{\projX m^* - \projX Y}}} \\ 
    =\;\; & \text{I} \;+\; \text{II} \; + \; \text{III}.
\end{aligned} 
\end{equation*}
It remains to bound the terms $\text{I}, \text{II}, \text{III}$. \\

\noindent \textbf{Bound on $\text{I}$:} We claim that $\text{I}\leq  \linearRisk{\B-1}{G, F}$. This holds since:
$$ \EE{\Norm{\mu - \projX m^*}_2^2 \cond \ordX} = \inf_{m \,\in\, \linclass} \EE{\Norm{\mu - m(\ordX)}_2^2 \cond \ordX}.$$
Thus:
$$
\begin{aligned}
\textrm{I} = \frac{1}{\n}\EE{\EE{\Norm{\mu - \projX m^*}_2^2 \cond \ordX}}
\leq  \inf_{m \,\in\, \linclass}\cb{ \frac{1}{\n} \EE{\Norm{\mu - m(\ordX)}_2^2}} 
= \linearRisk{\B-1}{G, F}.
\end{aligned}
$$

\noindent \textbf{Bound on $\text{II}$:}  
The argument here is similar to  results on fixed design linear regression, see e.g., Theorem 11.1 in~\citet{gyorfi2006distribution}. We have that:
$$
\begin{aligned}
\EE{\Norm{\projX m^* -\projX Y}_2^2}
&= \EE{ \Norm{ \projX\p{m^*(\ordX) - Y}}^2_2} \\
&= \EE{ \tr\p{ \projX \p{m^*(\ordX) - Y}\p{m^*(\ordX) - Y}^\top}} \\
&= \EE{  \tr\p{ \projX \EE{\p{m^*(\ordX) - Y}\p{m^*(\ordX) - Y}^\top \cond \ordX}}} \\
&= \EE{ \tr\p{ \projX \Var{Y \cond \ordX}}},
\end{aligned}
$$
where $\VarInline{Y \cond \ordX}$ is the $\n \times \n$ diagonal matrix with $i$-th diagonal entry equal to $\VarInline{Y_i \cond \ordX_i}$.  To conclude with the bound on  \smash{$\EEInline{\Norm{\projX m^* -\projX Y}_2^2}$}, we provide two separate bounds on \smash{$\EEInline{ \tr( \projX \VarInline{Y \cond \ordX})}$} according to whether Assumption (i) or (ii) of Theorem~\ref{theo:lin_regret} holds.

We start with (i). We note that $\VarInline{Y \cond \ordX} \preceq \max_i\{\VarInline{Y_i \cond \ordX_i}\} \cdot I_{\n}$ in the positive semi-definite order. By trace properties and since $\projX$ is an orthogonal projection matrix, it thus holds that:
$$\tr\p{ \projX \VarInline{Y \cond \ordX}} \leq \max_{i=1,\dotsc,\n} \VarInline{Y_i \cond \ordX_i} \cdot \B.$$ 
Integrating over $\ordX$ and using assumption (i), we conclude that:
$$\EE{\tr\p{ \projX \VarInline{Y \cond \ordX}}} \leq C_{\n} \cdot \B \;\; \Longrightarrow \;\; \text{II} \leq \frac{C_{\n} \cdot \B}{\n}.$$
We turn to (ii). Let $h_i = (\projX)_{ii}$, then:
$$
\begin{aligned}
\EE{ \tr\p{ \projX \VarInline{Y \cond \ordX}}} &= \sum_{i=1}^{\n} \EE{ h_i \VarInline{Y_i \cond \ordX_i}} \\
& \leq \sum_{i=1}^{\n} \EE{h_i^2}^{1/2} \EE{ \VarInline{Y_i \cond \ordX_i}^2}^{1/2} \\
& \leq \n \sqrt{\frac{\B}{\n}} \Gamma.
\end{aligned}
$$
For the penultimate inequality we used Cauchy–Schwarz and for the last inequality, we used the following two intermediate results: first, \smash{$\EEInline{\sum_{i=1}^{\n} h_i} = \EEInline{\tr(\projX)} \leq \B$} and so by symmetry $\EE{h_i} \leq \B/\n$ and $\EE{h_i^2} \leq \EE{h_i}\leq \B/\n$. Second,
$$\EE{ \VarInline{Y_i \cond \ordX_i}^2} \leq \EE{ \EE{Y_i^2 \cond \ordX_i}^2} \leq \EE{\EE{Y_i^4 \cond \ordX_i}} = \EE{Y_i^4} \leq \Gamma^2.$$
\noindent \textbf{Bound on $\text{III}$:} It remains to bound the cross-term $\text{III}$. We will show that $\text{III} \leq 0$. The crucial fact we use is that \smash{$\EEInline{Y \cond \mu, \ordX} = \EEInline{Y \cond \mu} = \mu$}:
\begin{equation*}
\begin{aligned}
\EE{ \p{\mu - \projX m^*}^\top\p{\projX m^* - \projX Y} \cond \ordX, \mu } &= \EE{ \p{\mu - \projX m^*}^\top\projX \p{m^* - Y} \cond \ordX, \mu} \\
&=  \p{\mu - \projX m^*}^\top\projX \p{m^* - \mu}  \\
&\stackrel{(*)}{=}  \p{\mu - m^*}^\top\projX \p{m^* - \mu} \\
&\stackrel{(**)}{\leq} 0.
\end{aligned}
\end{equation*}
Integrating over $\ordX, \mu$ we conclude. Let us  justify $(*), (**)$.\\
$(*)$ follows because $\p{m^* - \projX m^*}^\top\projX \p{m^* - \mu}=0$, since $\p{m^* - \projX m^*}$ is orthogonal to the space spanned by $\projX$. $(**)$ holds since for any vector $v$, it holds that $v^\top \projX v \geq 0$ by positive semidefiniteness of $\projX$.

The claim at the end of Theorem~\ref{theo:lin_regret} follows directly by Theorem~\ref{theo:eb_decomp} and the bounds we derived above for $\text{II}$, since:
$$\Err\p{m^*, \hat{m}} = \text{II} = \frac{1}{\n} \EE{\Norm{\projX m^* -\projX Y}_2^2}.$$
\end{proof}

\section{Normal prior and Homoskedastic Normal likelihood}
\label{sec:example_details}

We consider the more general case where $Z_{ij} \cond \mu_i \sim \mathcal{N}(\mu_i, \sigma^2)$. The results for the example may be recovered by taking $\sigma^2=1$. Throughout we assume that $\sigma^2, A > 0$.

\subsection{Oracle risks: Details for Example~\ref{exam:normal_normal_regret}} 
\label{subsec:oracle_risks}

The three oracle estimators are:
$$\EE{\mu_i \cond \boldZ_i} = \frac{A}{A+\sigma^2/\B} \bar{Z}_i,\;\; m^*(\ordX_i) = \frac{A}{A+\sigma^2/(\B-1)} \bar{X}_i,\;\;\AvgBayesRule(\boldZ_i) = \frac{A}{A+\sigma^2/(\B-1)} \bar{Z}_i,$$
where $\bar{Z}_i$ is the sample average of $Z_{i1},\dotsc,Z_{i\B}$.
They have the following oracle risks:
\begin{alignat*}{2}
&\bayesRisk{\B}{G,F} &&= \frac{A \sigma^2}{AK + \sigma^2}\\
&\bayesRisk{\B-1}{G,F} &&= \frac{A \sigma^2}{A(K-1) + \sigma^2}\\
& \AvgBayesRisk{\B-1}{G,F} &&= \frac{A \sigma^2}{AK + \sigma^2} + \p{\frac{A\sigma^2}{ (A\B+\sigma^2)(A(\B-1) + \sigma^2)}}^2\p{A + \frac{\sigma^2}{\B}}
\end{alignat*}
Furthermore, the RHS of Proposition~\ref{prop:averaged_jackknife}, equation~\eqref{eq:threebenchmarks}, holds with equality, i.e.,
$$
\begin{aligned}
&\AvgBayesRisk{\B-1}{G,F} = \frac{A \sigma^2}{A(K-1) + \sigma^2} - \EE{\VarInline{   m^*(\ordX_i) \cond \mu_i}}/\B,\\
\text{where }\;\; &\EE{\VarInline{   m^*(\ordX_i) \cond \mu_i}} = \p{\frac{A}{A+\sigma^2/(\B-1)}}^2  \cdot \frac{\sigma^2}{\B-1}.
\end{aligned}
$$

\begin{proof}
The first two calculations are standard, and we provide the details for $\AvgBayesRule$ and its risk only. First, by definition of $\AvgBayesRule$:
$$ \AvgBayesRule(\boldZ_i) = \frac{1}{\B}\sum_{j=1}^{\B}m^*\p{\ordX_i(j)} =  \frac{1}{\B}\sum_{j=1}^{\B}\frac{A}{A+\sigma^2/(\B-1)} \bar{X}_i(j) = \frac{A}{A+\sigma^2/(\B-1)} \bar{Z}_i.$$
To evaluate the risk, let us write $\lambda = A/(A+\sigma^2/(\B-1))$ and $\lambda^*=A/(A+\sigma^2/\B)$ (recall $\lambda^* \bar{Z}_i$ is the Bayes rule here). Then:

$$
\begin{aligned}
\EE{(\mu_i - \lambda \bar{Z}_i)^2} & = \EE{\p{(\mu_i  - \lambda^* \bar{Z}_i) + (\lambda^*-\lambda)\bar{Z}_i}^2} \\
&= \EE{\p{(\mu_i  - \lambda^* \bar{Z}_i}^2} + (\lambda^*-\lambda)^2 \EE{\bar{Z}_i^2} \\
&=\bayesRisk{\B}{G,F} + (\lambda^*-\lambda)^2 \p{A + \frac{\sigma^2}{\B}} \\ 
&= \bayesRisk{\B}{G,F}  + \p{\frac{A\sigma^2}{ (A\B+\sigma^2)(A(\B-1) + \sigma^2)}}^2\p{A + \frac{\sigma^2}{\B}}.
\end{aligned}
$$
Furthermore:
$$
\begin{aligned}
\EE{\VarInline{   m^*(\ordX_i) \cond \mu_i}}/\B &= \EE{ \p{\frac{A}{A+\sigma^2/(\B-1)}}^2 \Var{ \bar{X}_i \cond  \mu_i}}/\B \\ 
& = \p{\frac{A}{A+\sigma^2/(\B-1)}}^2  \cdot \frac{\sigma^2}{\B(\B-1)}.
\end{aligned}
$$

\end{proof}
\subsection{Data-driven risks: Details for  Example~\ref{exam:normal_normal}}

Recall that the Bayes risk $\bayesRisk{\B}{G,F}$ is equal to $\p{A \sigma^2}/\p{AK + \sigma^2}$. The upper bounds on the risk of the three estimators considered in this example are as follows:.

\paragraph{Auroral:} 
$$
\begin{aligned}
\frac{1}{\n}\sum_{i=1}^{\n} \EE{\p{ \mu_i - \muebdsl_i}^2} &\leq \bayesRisk{\B}{G, F} \, + \, 2\p{\AvgBayesRisk{\B-1}{G, F} - \bayesRisk{\B}{G, F}} \, +\, \frac{\B}{\n}\cdot\p{\sigma^2 + \frac{A \frac{\sigma^2}{\B-1}}{A+\frac{\sigma^2}{\B-1}}} \\
&= \bayesRisk{\B}{G, F} \, + \, O\p{ \frac{1}{\B^4}} \,+\, O\p{ \frac{\B}{\n}}.
\end{aligned}
$$
\paragraph{CC-L:} 
$$
\begin{aligned}
\frac{1}{\n}\sum_{i=1}^{\n} \EE{\p{ \mu_i - \hat{\mu}^{\text{CC-L}}}^2} &\leq \bayesRisk{\B}{G, F} \, + \, 2\p{\AvgBayesRisk{\B-1}{G, F} - \bayesRisk{\B}{G, F}} \, +\, \frac{2}{\n}\cdot\p{\sigma^2 + \frac{A \frac{\sigma^2}{\B-1}}{A+\frac{\sigma^2}{\B-1}}} \\
&= \bayesRisk{\B}{G, F} \, + \, O\p{ \frac{1}{\B^4}} \,+\, O\p{ \frac{1}{\n}}.
\end{aligned}
$$
\paragraph{James-Stein:} 

\begin{flalign*}
\frac{1}{\n}\sum_{i=1}^{\n} \EE{\p{ \mu_i - \hat{\mu}^{\text{JS}}_i}^2}&=  \bayesRisk{\B}{G, F} \; + \; \frac{1}{\n\B^2}\frac{2\B \sigma^4}{ \sigma^2 + A\B} && \\
&= \bayesRisk{\B}{G, F} \, + \, O\p{ \frac{1}{\n \B^2}}. &&
\end{flalign*}

\begin{proof}
We consider each estimator separately.\\

\noindent \textbf{Auroral:} The result follows from Theorem~\ref{theo:eb_decomp} along with the bound on $\Err\p{m^*, \hat{m}}$ derived in Theorem~\ref{theo:lin_regret}. To apply the latter, we need to compute $\VarInline{ Y_i \mid \ordX_i}$.

$$ \VarInline{ Y_i \mid \ordX_i} = \Var{ Y_i \mid \bX_i}=\Var{Y_i} - \frac{\Cov{Y_i, \widebar{X}_i}^2}{\VarInline{\widebar{X}_i}} = \sigma^2 +A - \frac{A^2}{A+\frac{\sigma^2}{\B-1}}.$$
Thus:
\begin{equation}
\label{eq:normal_normal_var}
\VarInline{ Y_i \mid \ordX_i}= \sigma^2 + \frac{A \frac{\sigma^2}{\B-1}}{A+\frac{\sigma^2}{\B-1}}.
\end{equation}

\noindent \textbf{CC-L:}  The derivation is similar to the result for Auroral above and is omitted. The main difference is that we solve a linear least squares problem with an intercept and a single regressor (instead of an intercept and $\B-1$ regressors). \\

\noindent \textbf{James-Stein:} The calculations are standard, e.g., see Chapter 1 of~\citet{efron2012large}.  First, using Stein's identity we get that:
$$ \EE{\Norm{\hat{\mu}^{JS} - \mu}_2^2} = \n\frac{\sigma^2}{\B} - \frac{\sigma^4}{\B^2} (\n-2)^2 \EE{\frac{1}{\Norm{\bZ}^2_2}}.$$
Now notice that $\Norm{\bZ}^2_2 \sim (\sigma^2/\B + A)\chi^2_{\n}$, where $\chi^2_{\n}$ is the $\chi^2$ distribution with $\n$ degrees of freedom, and so $\EE{(\n-2)(\sigma^2/\B+A)/\Norm{\bZ}^2_2} = 1$, which in turn yields:

$$ \EE{\Norm{\hat{\mu}^{JS} - \mu}_2^2} = \n\frac{\sigma^2}{\B} - \frac{\frac{\sigma^4}{\B^2} (\n-2)}{\frac{\sigma^2}{\B} +A} = \frac{\n \frac{\sigma^2}{\B} A}{\frac{\sigma^2}{\B} + A} + \frac{2\frac{\sigma^4}{\B^2}}{
\frac{\sigma^2}{\B} +A }.$$
Dividing by $\n$ we get the required claim.

\end{proof}

\subsubsection{Auroral and CC-L with a single held-out response}
We conclude this section by noting that we can compute the risk of Auroral and CC-L exactly, if we omit the averaging over $j$ in the Aurora algorithm. In particular:
$$
\begin{aligned}
&\frac{1}{\n}\sum_{i=1}^{\n} \EE{\p{ \mu_i - \muebdsl_{i,j}}^2}&=\;\; &\bayesRisk{\B-1}{G, F} \; &+ \;\frac{\B}{\n}\p{\sigma^2 - \frac{A \sigma^2}{\sigma^2 + A(\B-1)}}. \\
&\frac{1}{\n}\sum_{i=1}^{\n} \EE{\p{ \mu_i - \hat{\mu}^{\text{CC-L}}_{i,j}}^2}&=\;\; &\bayesRisk{\B-1}{G, F} \; &+ \; \frac{2}{\n}\p{\sigma^2 - \frac{A\sigma^2}{\sigma^2+ A(\B-1)}}.
\end{aligned}
$$

\begin{proof}

We start with Auroral. Throughout we use the same notation as in Section~\ref{sec:auroral_proof}. Further, we write $\hat{\mu}$ for $\muebdsl_{\cdot,j} = \projXj Y(j) = \projX Y$. First:

$$
\begin{aligned}
\EE{ \Norm{ \mu - \hat{\mu}}^2_2} &= \EE{ \Norm{ \mu - Y + Y- \hat{\mu}}^2_2} \\
&= \EE{ \Norm{Y-\hat{\mu}}^2_2} + \EE{\Norm{ Y-\mu}_2^2} - 2 \EE{(Y-\mu)^\top \p{Y-\hat{\mu}}}.
\end{aligned}
$$
Notice that $\EE{\Norm{ Y-\mu}_2^2} = \n\sigma^2$. By conditioning on $\mu, \ordX$, we get:
$$
\begin{aligned}
\EE{(Y-\mu)^\top \p{Y-\hat{\mu}}} &= \EE{(Y-\mu)^\top (I-\projX)Y} \\
&= \EE{ (Y-\mu)^\top (I-\projX)(Y-\mu)} \\
&= \EE{ \tr\p{ (I-\projX)\EE{(Y-\mu)(Y-\mu)^\top \cond \mu, \ordX}}} \\
& = (\n-\B)\sigma^2.
\end{aligned}
$$   
Here we used that $Y_i \mid \mu_i$ is homoskedastic. So as an intermediate result we conclude that:
$$\EEInline{ \Norm{ \mu - \hat{\mu}}^2_2} = \EEInline{ \Norm{Y-\hat{\mu}}^2_2} - (\n-2\B)\sigma^2.$$
For the remaining term we condition on $\ordX$ and observe that $\projX \EEInline{Y \cond \ordX} = \EEInline{Y \cond \ordX}$ and that $Y_i \mid \ordX_i$ is homoskedastic:
$$
\begin{aligned}
\EE{ \Norm{Y- \hat{\mu}}^2_2 \cond \ordX} &= \EE{ \Norm{ (I-\projX)Y}^2_2 \cond \ordX} \\
&=\EE{ \Norm{ (I-\projX)\p{Y - \EEInline{Y \cond \ordX}}}^2_2 \cond \ordX} \\ 
&= \EE{ \tr\p{ (I-\projX) \VarInline{Y \cond \ordX}} \cond \ordX} \\
&= (\n-\B)\VarInline{Y_1 \cond \ordX_1}.
\end{aligned}
$$
We conclude by using the expression for $\VarInline{ Y_i \mid \ordX_i}$ derived in~\eqref{eq:normal_normal_var}, by iterated expectation and by rearranging terms.  The result for CC-L can be derived in a similar way. We note that the formula for CC-L appears (with a typo) as Corollary 1 in~\citet{coey2019improving}.

\end{proof}

\section{Location families: Formal statements for Section~\ref{subsec:loc_family}}
\label{sec:loc_family_example}

\subsection{Regularity assumptions}
\label{subsec:regularity}

\begin{assum}[Regular prior]
\label{assum:regular_prior}
The prior distribution $G$ has compact support $[t_1,t_2]$, $t_1 < t_2$, and has density $g$ w.r.t.\ the Lebesgue measure. Furthermore, $g$ is absolutely continuous on $[t_1,t_2]$, satisfies $g(t_1)=g(t_2)=0$ and also has finite Fisher information:
$$\mathcal{I}(g) := \int \frac{ g'(x)^2}{g(x)}\ind(g(x)>0)dx \, < \,\infty.$$
\end{assum}

\begin{assum}[Regular location density]
\label{assum:regular_loc}
We assume that $F(\cdot \mid \mu_i)$ has Lebesgue density $f( \cdot - \mu_i)$ with $f(\cdot)$ a fixed density and write (with some abuse of notation) $F(t) = \int_{(-\infty,t]} f(x) dx$ for the corresponding distribution. We assume that:
\begin{enumerate}[(i), leftmargin=*]
\item $f(\cdot)$ is symmetric around $0$.
\item The fourth moment of $f$ exists, i.e., $\int x^4 f(x)dx < \infty$. 
\item $f(\cdot)$ is twice continuously differentiable and $\cb{x \mid f(x) >0} \subset \mathbb R$ is an interval. The function $u \mapsto  \ell''(F^{-1}(u))$ is $(1/2 + \delta)$-Hölder continuous on $(0,1)$ for some $\delta > 0$,\footnote{
That is, there exist $L, \delta > 0$ such that $\abs{ \ell''(F^{-1}(u)) -  \ell''(F^{-1}(u'))} \leq L \abs{u-u'}^{1/2 + \delta}$ for all $u,u' \in (0,1)$.
}
where $\ell(x) = \log(f(x))$. The location Fisher information exists and is finite:
$$ \mathcal{I}(f) := \int \frac{ f'(x)^2}{f(x)}\ind(f(x)>0)dx = \int \ell'(x)^2 f(x)\ind(f(x)>0) dx \,< \,\infty.$$
\end{enumerate}
\end{assum}

\begin{rema}
Two examples of densities $f(\cdot)$ satisfying Assumption~\ref{assum:regular_loc} are:
\begin{enumerate}
    \item The Gaussian density $f(x)=\exp(-x^2/(2\sigma^2))/\sqrt{2\pi \sigma^2}$ for $\sigma >0$.   In this case, $\ell''(F^{-1}(u)) = -1/\sigma^2$ and $\mathcal{I}(f)=1/\sigma^2$.
    \item The logistic density, $f(x) = \exp(-x/s)/\{s \p{1 + \exp(-x/s)}^2\}$ for $s >0$. In this case, $\ell''(F^{-1}(u)) = -2u(1-u)/s^2$ and $\mathcal{I}(f)=1/(3s^2)$.
\end{enumerate}
The Laplace density does not satisfy the above regularity conditions. However, we may manually check that the conclusion of Corollary~\ref{coro:smoothloc} holds in that case as well (see Remark~\ref{rema:laplace_proof} below for the proof in the Laplace case).
\end{rema}

\subsection{Proof of Corollary~\ref{coro:smoothloc}}

We split up the proof over 3 individual parts, for each of the possible estimators considered.

\subsubsection{Proof for CC-L estimator}
\label{subsubsec:ccl_loc_proof}
\begin{proof}
To simplify the exposition of the proof, we make a centering assumption that $\EE[G]{\mu}=0$ and we fit the linear regression without intercept ($\beta_0 = 0$). The results are identical for the uncentered case where $\EE[G]{\mu} \neq 0$ and we use an intercept.

\noindent \textbf{Lower bound:} The estimator takes the form, 
$$\muccl_i =  \frac{1}{\B}\sum_{j=1}^{\B}\hat{\beta}(j)\bar{X}_i(j),\;\;\; \bar{X}_i(j) = \frac{1}{\B-1}\sum_{1 \leq \ell \neq j \leq \B} Z_{i\ell},\;\;\;
\hat{\beta}(j) = \frac{\sum_{i=1}^{\n} Y_i(j) \bar{X}_i(j)}{\sum_{i=1}^{\n} \bar{X}_i(j)^2}.$$
Let us write:
\begin{equation}
\begin{aligned}
\sqrt{\B}\p{\muccl_1 - \mu_1} &= \sqrt{\B}\p{\frac{1}{\B}\sum_{j=1}^{\B}\hat{\beta}(j)\bar{X}_1(j) - \mu_1} \\
&= \frac{1}{\sqrt{\B}}\sum_{j=1}^{\B}\p{\hat{\beta}(j)-1}\bar{X}_1(j) \;+\; \sqrt{\B}\p{\frac{1}{\B}\sum_{j=1}^{\B}\bar{X}_1(j) - \mu_1}\\ 
&\stackrel{(*)}{=} o_{\mathbb P}(1)\;  + \; \sqrt{\B}\p{\bar{Z}_1 - \mu_1 }
\end{aligned}
\end{equation}
We will justify $(*)$ below. From the Central Limit Theorem and Slutsky it follows that,
$$\sqrt{\B}\p{\muccl_1 - \mu_1} \xrightarrow{\mathcal{D}} \mathcal{N}\p{0, \sigma^2}.$$
Then, by Fatou's Lemma:
$$ \liminf_{\n \to \infty}\EE{\p{\sqrt{\B}(\muccl_1 - \mu_1)}^2} \geq \EE{\tilde{Z}^2} = \sigma^2, \text{ where } \tilde{Z}\sim \mathcal{N}\p{0, \sigma^2}.$$ 
By exchangeability: $\frac{1}{\n}\sum_{i=1}^{\n} \EE{\p{ \mu_i - \hat{\mu}^{\text{CC-L}}_i}^2} = \EE{ (\muccl_1 - \mu_1)^2}$ and so we have established that:
$$ \liminf_{\n \to \infty}\cb{ \frac{1}{\n}\sum_{i=1}^{\n} \EE{\p{ \mu_i - \hat{\mu}^{\text{CC-L}}_i}^2}\bigg/ \frac{\sigma^2}{\B} \geq 1}.$$ 
It remains to show $(*)$. It suffices to show that:
\begin{equation}
\label{eq:ccl_analysis}
\max_{j=1}^{\B}\cb{\abs{\bar{X}_1(j)}} = O_{\mathbb P}(1) \; \text{ and }\;  \sqrt{\B}\max_{j=1}^{\B}\cb{\abs{ \hat{\beta}(j)-1}} =  o_{\mathbb P}(1).
\end{equation}
For the first of these results, note that:
$$
\begin{aligned}
\max_{j=1}^{\B}\cb{\abs{\bar{X}_1(j)}} &= \max_{j=1}^{\B}\cb{\abs{\frac{\B}{\B-1}\bar{Z}_1 - \frac{Z_{1j}}{\B-1} }}\\
& \leq  \frac{\B}{\B-1}\p{ \abs{\bar{Z}_1} + \max_{j=1}^{\B}\cb{\frac{\abs{Z_{1j}}}{\B}}}.  
\end{aligned}
$$
By the law of large numbers, it holds that $\abs{\bar{Z}_1} = O_{\mathbb P}(1)$. Furthermore,  $\max_{j=1}^{\B}\cb{\abs{Z_{1j}}/{\B}} = o_{\mathbb P}(1)$, since for any $\varepsilon > 0$
$$
\begin{aligned}
\PP{ \max_{j=1}^{\B}\cb{\frac{\abs{Z_{1j}}}{\B}} > \varepsilon} &\leq \sum_{j=1}^{\B}\PP{\abs{Z_{1j}} > \varepsilon \B}\leq \B \frac{ \EE{ Z_{1j}^2}}{\varepsilon^2 \B^2} \to 0 \text{ as } \B \to \infty.
\end{aligned}
$$
For the second part of~\eqref{eq:ccl_analysis}, let us note that $\EE{Y_i(j) \bar{X}_i(j)} = A$ and $\EE{\bar{X}_i(j)^2} = A + \sigma^2/(\B-1)$, where $A = \int \mu^2 g(\mu)d\mu$. As a first consequence we have that $\EE{M_i(j)}=0$, where:
$$ M_i(j) := Y_i(j) \bar{X}_i(j) - \frac{A}{A + \sigma^2/(\B-1)}\bar{X}_i(j)^2.$$
Second, for any $\varepsilon' >0$ and any $j$, we get that:
$$
\begin{aligned}
 \PP{\abs{\frac{1}{\n} \sum_{j=1}^{\n}  \frac{\bar{X}_i(j)^2}{\EE{\bar{X}_i(j)^2}} - 1} > \varepsilon'} \leq  \frac{ \EE{\bar{X}_i(j)^4 }}{\n (\varepsilon')^2\EE{\bar{X}_i(j)^2}^2}.
\end{aligned} 
$$
Next, fixing another $\varepsilon>0$, it holds that:
$$
\begin{aligned}
& \PP{\sqrt{\B} \max_{j=1}^{\B}\cb{\abs{ \hat{\beta}(j)- \frac{A}{A + \sigma^2/(\B-1)}}} > \varepsilon} \\
\leq \;\; &\sum_{j=1}^{\B} \PP{\abs{ \hat{\beta}(j)- \frac{A}{A + \sigma^2/(\B-1)}} > \varepsilon/\sqrt{\B}} \\ 
\leq \;\; &\B \cdot  \PP{\abs{ \frac{1}{\n}\sum_{i=1}^{\n} M_i(j)} > \varepsilon/\sqrt{\B} \cdot \frac{1}{\n} \sum_{i=1}^n \bar{X}_i(j)^2} \\
\leq \;\; & \B \cdot \cb{ \PP{\abs{ \frac{1}{\n}\sum_{i=1}^{\n} M_i(j)} > \varepsilon/\sqrt{\B} \cdot \EE{\bar{X}_i(j)^2}/2} \, + \,  \PP{\frac{1}{\n} \sum_{j=1}^{\n} \bar{X}_i(j)^2 \leq \EE{\bar{X}_i(j)^2}/2}}\\ 
\leq \;\; & 4\B \cdot \cb{ \frac{  \B \Var{M_i(j)}}{\n \varepsilon^2 \EE{\bar{X}_i(j)^2}^2 } \, + \, \frac{ \EE{\bar{X}_i(j)^4 }}{\n \EE{\bar{X}_i(j)^2}^2}} \, \to 0 \text{ as } \n \to \infty,
\end{aligned}
$$
since $\B^2/\n \to 0$. Thus:
$$\sqrt{\B}\max_{j=1}^{\B}\cb{\abs{ \hat{\beta}(j)-\frac{A}{A + \sigma^2/(\B-1)}}} =  o_{\mathbb P}(1).$$
We conclude by noting that
$$\sqrt{\B}\p{\frac{A}{A+\sigma^2/(\B-1)} - 1} = \frac{\sqrt{\B}}{\B-1}\cdot\frac{\sigma^2}{A+\sigma^2/(\B-1)} = o(1).$$

\noindent\textbf{Upper bound:} Repeating the argument of the proof of Theorem~\ref{theo:lin_regret} and noting that $\Var{Y_i \cond \bar{X}_i} \leq C$ for some $C< \infty$ by our assumptions on the support of $\mu$ and since $\sigma_i^2 = \Var{Y_i \cond \mu_i}$ is constant for location families, we get:
$$ \frac{1}{\n}\sum_{i=1}^{\n} \EE{\p{\mu_i - \muccl_i}^2} \leq \inf_{\beta \in \mathbb R} \EE{\p{\mu_i - \beta \cdot \bar{X}_i}^2} + C \frac{1}{\n}.$$
The first term is upper bounded by $\sigma^2/(\B-1)$ by plugging in $\beta = 1$ and since $\B/\n \to 0$ we get:

$$ \limsup_{\n \to \infty} \frac{1}{\n}\sum_{i=1}^{\n} \EE{\p{ \mu_i - \hat{\mu}^{\text{CC-L}}_i}^2}\bigg/ \frac{\sigma^2}{\B} \leq 1. $$
We conclude, that
 $$\frac{1}{\n}\sum_{i=1}^{\n} \EE{\p{ \mu_i - \hat{\mu}^{\text{CC-L}}_i}^2} \bigg/ \frac{\sigma^2}{\B} \to 1 \text { as } \n \to \infty.$$
\end{proof}

\subsubsection{Proof for $\EE{\mu_i \mid \bar{Z}_i}$ estimator}
\label{subsubsec:averaged_post_mean_loc}
\begin{proof}
\textbf{Lower bound:} Let $\sigma^2 = \Var{Z_i \cond \mu_i}$ (it does not depend on $\mu_i$ for location families) and define the normalized sum \smash{$U_{i,\B} = \sum_{j=1}^{\B} Z_{ij} / \sqrt{\B \sigma^2}$}. Since $\sigma^2 >0$ and $\mathcal{I}(f) < \infty$, \citet[Theorem 1.6]{johnson2004fisher} prove that:
$$\mathcal{I}(U_{i,\B}) \to 1 \text{ as } \B \to \infty,$$
where $\mathcal{I}(U_{i,\B})$ is --- with some abuse of notation --- the location Fisher information of the Lebesgue density of $U_{i,\B}$ conditionally on the location parameter $\mu_i$.\footnote{Part of the statement is also that $\mathcal{I}(U_{i,\B})$ exists and is finite for all $\B$ large enough.} Next, note that $\bar{Z}_i = U_{i,\B}\sigma / \sqrt{\B}$ and so $\mathcal{I}(\bar{Z}_i) = \B \mathcal{I}(U_{i,\B})/\sigma^2$. By Corollary~\ref{coro:van_trees} below (van Trees inequality), applied for $\bar{Z}_i$ and one replicate, 
$$ \EE{ \p{ \mu_i - \EE{\mu_i \cond \bar{Z}_i}}^2} \geq \frac{1}{\mathcal{I}(\bar{Z}_i) + \mathcal{I}(g)} = \frac{1}{\B/\sigma^2(1+o(1)) + \mathcal{I}(g)}.$$
\textbf{Upper bound:} The Bayes risk of $\EE{\mu_i \cond \bar{Z}_i}$ is upper bounded by the risk of $\bar{Z}_i$, i.e.,
$$ \EE{ \p{ \mu_i - \EE{\mu_i \cond \bar{Z}_i}}^2} \leq \Var{\bar{Z}_i} = \frac{\sigma^2}{\B}.$$
Combining the upper and lower bounds, we find that,
$$\B\EE{ \p{ \mu_i - \EE{\mu_i \cond \bar{Z}_i}}^2} \to \sigma^2 \text{ as } \B \to \infty.$$
\end{proof}

\subsubsection{Proof for Auroral estimator}

\begin{proof}
\textbf{Lower bound:} By Corollary~\ref{coro:van_trees} below (van Trees inequality) it holds for the Bayes risk that,
$$ \bayesRisk{\B}{G,F} \geq \frac{1}{\B \mathcal{I}(f) + \mathcal{I}(g)}.$$
So:
$$ \liminf_{\n \to \infty} \frac{\B}{\n}\sum_{i=1}^{\n} \EE{\p{ \mu_i - \muebdsl_i}^2} \geq \liminf_{\n \to \infty} \p{\B \bayesRisk{\B}{G,F}} \geq \frac{1}{\mathcal{I}(f)}.$$
\textbf{Upper bound:}  The first step consists of upper bounding the risk $\linearRisk{\B-1}{G,F}$. We will do this by choosing an appropriate function from the class \smash{$\text{Lin}\p{\mathbb R^{\B-1}}$}, cf.~\eqref{eq:linear_class} and bounding its mean squared error for estimating $\mu_i$. To this end, fix $\beta_{(0)}$ for the intercept and let \smash{$\beta_{(j)} = h(j/\B)$},  $j=1,\dotsc, \B-1$ for a $(1/2+\delta)$-Hölder continuous function $h:[0,1] \to \RR$ to be chosen below. The L-statistic we study takes the form
\begin{equation}
\label{eq:Lstatistic_general}
T_i = \beta_{(0)} + \frac{1}{\B-1} \sum_{j=1}^{\B-1} h(j/\B) X_{i}^{(j)}.
\end{equation}
We may write $X_{ij} = \mu_i + \tilde{X}_{ij}$ where $\tilde{X}_{ij} \sim f(\cdot)$ is $0$-centered, so that:
$$ T_i =  \mu_i \frac{1}{\B-1} \sum_{j=1}^{\B-1} h(j/\B) + \beta_{(0)} + \frac{1}{\B-1} \sum_{j=1}^{\B-1} h(j/\B) \tilde{X}_{i}^{(j)}.$$
Here $\tilde{X}_{i}^{(j)}  = X_{i}^{(j)}-\mu_i $. By Hölder continuity of $h$:
$$\frac{1}{\B-1}\sum_{j=1}^{\B-1} h(j/\B) = \int_{0}^1 h(u)du + O\p{\B^{-1/2-\delta}}.$$
Choosing $\beta_{(0)}$ to be $-\frac{1}{\B-1} \sum_{j=1}^{\B-1} h(j/\B) \EEInline{\tilde{X}_{i}^{(j)}}$, we thus get that (with the $O(\B^{-1/2-\delta})$ term being uniform over $\mu_i$ in the support of $G$)
$$\text{Bias}(T_i \cond \mu_i) = \EE{T_i \cond \mu_i} - \mu_i =  \p{\int_{0}^1 h(u)du - 1}\cdot \mu_i + O\p{\B^{-1/2-\delta}}.$$
\citet[Proof of Theorem 22.3]{van2000asymptotic} establishes that:
$$\Var{T_i \cond \mu_i} = \frac{1}{\B-1}\int\int h(F(x))h(F(y))\p{F(x \land y) - F(x)F(y)}dxdy  + o(1/\B),$$
where the $o(1/\B)$ does not depend on $\mu_i$.
Now let us make the concrete choice $h(u) = - \frac{1}{\mathcal{I}(f)}\ell''\p{F^{-1}(u)}$.  At the end of the proof, we will show that for this choice of $h$:
\begin{equation}
\label{eq:fisher_info_calc}
    \int_0^1 h(u)du=1,\;\;\;\;\; \int\int h(F(x))h(F(y))\p{F(x \land y) - F(x)F(y)}dxdy = \frac{1}{\mathcal{I}(f)}
\end{equation}
Thus we get that
$$
\EE{(T_i - \mu_i)^2} = \text{Bias}(T_i \cond \mu_i)^2 + \Var{T_i \cond \mu_i} =  \frac{1}{(\B-1)\mathcal{I}(f)} + o\p{\frac{1}{\B}},
$$
and so,
\begin{equation}
\label{eq:upper_bound_lin_risk_fisher_info}
\B\linearRisk{\B-1}{G,F} \leq \frac{1}{\mathcal{I}(f)}(1+o(1)) \text { as } K \to \infty.
\end{equation}
Therefore, by Theorem~\ref{theo:lin_regret} and noting that $\VarInline{Y_i \cond \ordX_i} \leq C$ for some $C<\infty$ (arguing as in the proof for CC-L), we get
$$ \frac{\B}{\n}\sum_{i=1}^{\n} \EE{\p{ \mu_i - \muebdsl_i}^2}\leq \B\p{\linearRisk{\B-1}{G,F} + C\frac{\B}{\n}}\leq  \frac{1}{\mathcal{I}(f)}(1+o(1)) + C\frac{\B^2}{\n}$$
The conclusion follows by taking $\n \to \infty$ since then $\B \to \infty$, $\B^2/\n \to 0$, i.e., we get that,

$$\frac{1}{\n}\sum_{i=1}^{\n} \EE{\p{ \mu_i - \muebdsl_i}^2} \bigg/ \frac{\mathcal{I}(f)^{-1}}{\B} \to 1 \text{ as } \n \to \infty$$ 
\textbf{Fisher information calculations:} It remains to prove~\eqref{eq:fisher_info_calc}. First
$$ \int_0^1 h(u)du = -\frac{1}{\mathcal{I}(f)} \int_0^1 \ell''\p{F^{-1}(u)}du =- \frac{1}{\mathcal{I}(f)}\int \ell''(u) f(u)du =  \frac{\mathcal{I}(f)}{\mathcal{I}(f)} = 1.$$ 
For the second term, we first note that for the given choice of $h$:
$$
\begin{aligned}
&\int\int h(F(x))h(F(y))\p{F(x \land y) - F(x)F(y)}dx\,dy \\
=\,\;\; &\frac{1}{\mathcal{I}(f)^2} \int\int \ell''(x) \ell''(y)\p{F(x \land y) - F(x)F(y)}dx\,dy.
\end{aligned}
$$ 
Thus we need to show that the numerator of the last expression is equal to $\mathcal{I}(f) = \EE{\ell'(X)^2}$, where $X \sim f(\cdot)$. To this end, first let $X'$ be an i.i.d. copy of $X$. Then, since $\EE{\ell'(X)}=0$, it follows that $\EE{ \ell'(X)^2} =\EE{ (\ell'(X)- \ell'(X'))^2}/2$. On the other hand, by absolute continuity of $\ell'$, we may write $X,X'$ almost surely:
$$\ell'(X) - \ell'(X') = \int \ell''(x)\p{\ind(X' \leq x) - \ind (X \leq x)}\, dx.$$
Applying the same argument formally with variable $y$ instead of $x$ and multiplying the results, we get:
$$\p{\ell'(X) - \ell'(X')}^2 = \int \ell''(x)\ell''(y)\p{\ind(X' \leq x) - \ind (X \leq x)}\p{\ind(X' \leq y) - \ind (X \leq y)}\, dx\,dy.$$
Assume momentarily that we can apply Fubini's theorem, then $\EE{\p{\ell'(X) - \ell'(X')}^2}/2$ is equal to,
$$
\begin{aligned}
\frac{1}{2} &\int \ell''(x)\ell''(y)\EE{\p{\ind(X' \leq x) - \ind (X \leq x)}\p{\ind(X' \leq y) - \ind (X \leq y)}}\, dx\,dy \\ 
=\,\;\; & \int \ell''(x)\ell''(y)\p{F(x \land y) - F(x)F(y)}dx\,dy,
\end{aligned}
$$
as claimed. To see why we may apply Fubini's theorem, note that under the claimed assumptions, there exists a constant $C$ such that $\abs{\ell''(x)} \leq C$ for almost all $x$ in the support of $f(\cdot)$ and so, $X, X'$ almost surely:
$$
\begin{aligned}
&\int \int \abs{\ell''(x)\ell''(y)\p{\ind(X' \leq x) - \ind (X \leq x)}\p{\ind(X' \leq y) - \ind (X \leq y)}} dx\,dy \\ 
\leq \,\;\; & C^2 \int \int \abs{\ind(X' \leq x) - \ind (X \leq x)}\abs{\ind(X' \leq y) - \ind (X \leq y)} dx\,dy \\ 
= \,\;\; & C^2 (X'-X)^2.
\end{aligned}
$$
The conclusion follows by Tonelli's theorem, since $\EE{X^2} = \EE{X'^2} < \infty$.

\end{proof}

\begin{rema}[Proof in the case of the Laplace distribution.]

\label{rema:laplace_proof}
Consider the Laplace location family with Lebesgue density $f(\cdot-\mu)$, where
$$ f(x) = \frac{1}{\sqrt{2}\sigma} \exp\p{ -\sqrt{2}\abs{x}/\sigma}.$$
The scale parameterization is such that $\sigma^2 = \int f^2(x)dx$. We note that $f(\cdot)$ is absolutely continuous and is differentiable in quadratic mean with Fisher Information equal to $\mathcal{I}(f) = 2/\sigma^2$~\citep[Example 12.2.4]{lehmann2006testing}. 

These properties of $f(\cdot)$ suffice for most steps of the proof of Corollary~\ref{coro:smoothloc}. The higher order smoothness is required only for deriving the upper bound for the Auroral estimator. In particular, the construction of a L-statistic in~\eqref{eq:Lstatistic_general} used to upper bound $\linearRisk{\B-1}{G,F}$ in~\eqref{eq:upper_bound_lin_risk_fisher_info}, no longer works. However, the following simple choice works instead: we choose $T_i$ as the median of $\boldX_i$. For simplicity we take limits over even $\B$, i.e., we assume that $\B=2b$ for $b \in \mathbb N_{\geq 1}$ and take $T_i = X_{i}^{(b)}$ and $b \to \infty$. \citet{chu1955moments} then prove that:
$$ \EE{ \p{T_i - \mu_i}^2} \bigg/ \p{\frac{1}{4 f^2(0) (\B-1)}} \to 1 \text{ as } b \to \infty.$$
But $1/(4 f^2(0)) = \sigma^2/2 = \mathcal{I}(f)^{-1}$ and so~\eqref{eq:upper_bound_lin_risk_fisher_info} holds, and the rest of the proof follows verbatim.

\end{rema}

\subsubsection{Proof of Corollary~\ref{coro:rectangular_loc}}

As in the proof of Corollary~\ref{coro:smoothloc}, we need to bound the risk of each of the three estimators separately. It will be convenient to do the following preliminary calculation. In the rectangular location family,
$$ \sigma^2 = \Var{Z_{ij} \cond \mu_i} = \frac{1}{2B} \int_{-B}^B z^2 dz = \frac{B^2}{3}.$$
\textbf{CC-L:} All steps of the proof in Section~\ref{subsubsec:ccl_loc_proof} go through, and so we get,
$$\frac{1}{\n}\sum_{i=1}^{\n} \EE{\p{ \mu_i - \hat{\mu}^{\text{CC-L}}_i}^2} \bigg/ \frac{B^2}{3\B} \to 1 \text { as } \n \to \infty.$$
\textbf{Posterior mean based on average, $\EE{\mu_i \mid \bar{Z}_i}$:} In this case the proof of Section~\ref{subsubsec:averaged_post_mean_loc} goes through, and we get the same expression as for CC-L above. There is one subtlety involved in applying~\citet[Theorem 1.6]{johnson2004fisher}: the Fisher information of the rectangular distribution is not well-defined. Nevertheless, the convolution of $f(\cdot)$ with itself, i.e., the density of $Z_{i1} + Z_{i2}$, is equal to the triangular density, which has finite Fisher information and is absolutely continuous. Thus, we may apply Theorem 1.6 of~\citet{johnson2004fisher}.\\
\newline
\textbf{Auroral:} Let $T_i = (X_i^{(1)} + X_i^{(\B-1)})/2$ be the sample midrange based on $\ordX_i$. $T_i$ is unbiased for $\mu_i$ and~\citet[Section 4]{rider1957midrange} showed that $\Var{T_i \cond \mu_i} = 2B^2/(\B(\B+1))$. Since $T_i  \in \linclass$, it follows that $\B(\B+1)\linearRisk{\B-1}{G,F} \leq 2B^2$. Applying Theorem~\ref{theo:lin_regret} and noting that $\VarInline{Y_i \cond \ordX_i} \leq C$ for some $C<\infty$, we get
$$ \frac{\B(\B+1)}{\n}\sum_{i=1}^{\n} \EE{\p{ \mu_i - \muebdsl_i}^2}\leq \B(\B+1)\p{\linearRisk{\B-1}{G,F} + C\frac{\B}{\n}}\leq  2B^2 + C\frac{\B^2(\B+1)}{\n}.$$
Since $\B^3/\n \to 0$, we conclude that:
$$ \limsup_{\n \to \infty} \frac{\B(\B+1)}{\n}\sum_{i=1}^{\n} \EE{\p{ \mu_i - \muebdsl_i}^2} \; \leq \;  2B^2.$$
We conclude with the statement of the corollary by combining the results for the three estimators.

\subsubsection{The van Trees inequality}
\begin{coro}[The van Trees inequality for location families]
\label{coro:van_trees}
In the setting of Section~\ref{subsec:loc_family}, assume that:
\begin{enumerate}[(a)]
\item  The density $f(\cdot)$ w.r.t.\ the Lebesgue measure is absolutely continuous and symmetric around $0$.
\item The location Fisher information for $f(\cdot)$ exists and is finite:
$$ \mathcal{I}(f) := \int \frac{ f'(x)^2}{f(x)}\ind(f(x)>0)dx \, < \infty. $$
\item The prior distribution $G$ has compact support $[t_1,t_2]$ and has density $g$ w.r.t.\ the Lebesgue measure. Furthermore, $g$ is absolutely continuous on $[t_1,t_2]$, satisfies $g(t_1)=g(t_2)=0$ and also has finite Fisher information:
$$\mathcal{I}(g) = \int \frac{ g'(x)^2}{g(x)}\ind(g(x)>0)dx \, < \infty.$$
Then, the Bayes risk $\bayesRisk{\B}{G,F}$ for estimating $\mu_i$ satisfies:
$$ \bayesRisk{\B}{G,F} \geq \frac{1}{\B \mathcal{I}(f) + \mathcal{I}(g)}.$$
\end{enumerate}
\end{coro}

\begin{proof}
This is a direct corollary of the van Trees inequality~\citep*{van1968detection, gill1995applications}. Concretely, here we apply the form presented in Theorem 2.13 in~\citet{tsybakov2008introduction}. There results are phrased more generally for models $p(x,t)$, where $p(\cdot, t)$ is the Lebesgue density for a fixed parameter value $t$. Here we have $p(x,t) = f(x-t)$, which simplifies results. Concretely, we quickly verify the three assumptions (i),(ii),(iii) of Theorem 2.13 in~\citet{tsybakov2008introduction}:

For (i), joint measurability of $p(x,t)$ in $(x,t)$ follows from absolute continuity of $f(\cdot)$. For $(ii)$ we use translation invariance of the Lebesgue measure to note that the Fisher information is constant as a function of the location parameter. We also use the fact that the Fisher Information in the experiment where we observe $\B$ independent replicates is equal to $\B$-times the Fisher information in the experiment with a single replicate. Finally, assumption (iii) in~\citet{tsybakov2008introduction} is identical to assumption (c).
\end{proof}

\clearpage

\end{appendix}

\end{document}